\documentclass[journal]{IEEEtran}
\usepackage{graphicx}
\usepackage{hyperref}
\hypersetup{colorlinks,linkcolor={blue},citecolor={blue},urlcolor={red}}
\usepackage{cite}
\usepackage{booktabs}
\usepackage{placeins}
\usepackage{graphics,subfigure}
\usepackage{amsmath,amsfonts}
\usepackage{amsthm}

\newtheorem{lemma}{Lemma}

\usepackage{stmaryrd} 

\newtheorem{theorem}{Theorem}
\usepackage{multirow}
\usepackage{array}
\usepackage{soul}
\usepackage{algorithm}
\usepackage{algorithm}
\usepackage{algpseudocode}

\usepackage{amssymb}
\usepackage{cases}
\usepackage{setspace,multirow}
\usepackage{booktabs}
\usepackage{float}
\usepackage{mathtools}
\usepackage{geometry}
\usepackage{xspace}
\usepackage{nccmath}

\graphicspath{{Images/}}
\ifCLASSINFOpdf
\else
\fi
\newcommand{\orcid}[1]{\href{https://orcid.org/#1}{\textcolor[HTML]{A6CE39}{\aiOrcid}}}

\usepackage{tikz,xcolor,hyperref}
\usepackage{cleveref}
\definecolor{lime}{HTML}{A6CE39}
\DeclareRobustCommand{\orcidicon}{%
   \begin{tikzpicture}
    \draw[lime, fill=lime] (0,0) 
    circle [radius=0.16] 
    node[white] {{\fontfamily{qag}\selectfont \tiny ID}};    \draw[white, fill=white] (-0.0625,0.095) 
    circle [radius=0.007];    \end{tikzpicture}
\hspace{-2mm}}
\foreach \x in {A, ..., Z}{%
   \expandafter\xdef\csname orcid\x\endcsname{\noexpand\href{https://orcid.org/\csname orcidauthor\x\endcsname}{\noexpand\orcidicon}}
    }

\begin{document}

\title{SecureDyn-FL: A Robust Privacy-Preserving Federated Learning Framework for Intrusion Detection in IoT Networks}

\author{Imtiaz Ali Soomro\orcidI{}, Hamood Ur Rehman\orcidH{}, Syed Jawad Hussain\orcidJ{}, Adeel Iqbal\orcidA{}, Waqas Khalid\orcidW{}, and Heejung Yu\orcidQ{}\\

\thanks{This work was supported by the National Research Foundation of Korea (NRF) grant funded by the Korea government (MSIT) (RS-2025-00514779, RS-2025-02303435), by the Information Technology Research Center (ITRC) support program (IITP-2023-RS-2022-00164800) supervised by the Institute for Information \& Communications Technology Planning \& Evaluation (IITP). \textit{(Imtiaz Ali Soomro and Adeel Iqbal contributed equally to this work.)}. (\textit{Corresponding authors: Heejung Yu and Waqas Khalid}).}

\thanks{I. A. Somroo, and S. J. Hussain are with Sir Syed CASE Institute of Technology, Islamabad, Pakistan (email: imtiaz.soomro@case.edu.pk, jawad.hussain@case.edu.pk).}
\thanks{H. Khan is with the ECE Department, Habib University, Karachi, Pakistan (email: hamood.rehman@sse.habib.edu.pk).}
\thanks{A. Iqbal is with the School of Computer Science and Engineering, Yeunganam University, South Korea (email:adeeliqbal@yu.ac.kr).}
\thanks{Waqas Khalid is with the Department of Electrical and Electronic Engineering, and the Next Generation Internet of Everything Laboratory (NGIoE Lab), University of Nottingham Ningbo China, Ningbo 315100, China (e-mail: Waqas.Khalid@nottingham.edu.cn).}

\thanks{Heejung Yu is with the Department of Electronics and Information Engineering, Korea University, Sejong, 30019, South Korea (email: heejungyu@korea.ac.kr).}


}

\maketitle

\begin{abstract}
The rapid proliferation of Internet of Things (IoT) devices across domains such as smart homes, industrial control systems, and healthcare networks has significantly expanded the attack surface for cyber threats, including botnet-driven distributed denial-of-service (DDoS), malware injection, and data exfiltration. Conventional intrusion detection systems (IDS) face critical challenges like privacy, scalability, and robustness when applied in such heterogeneous IoT environments. To address these issues, we propose SecureDyn-FL, a comprehensive and robust privacy-preserving federated learning (FL) framework tailored for intrusion detection in IoT networks. SecureDyn-FL is designed to simultaneously address multiple security dimensions in FL-based IDS: (1) poisoning detection through dynamic temporal gradient auditing, (2) privacy protection against inference and eavesdropping attacks through secure aggregation, and (3) adaptation to heterogeneous non-independent-and-identically-distributed (non-IID) data via personalized learning. The framework introduces three core contributions: (i) a dynamic temporal gradient auditing mechanism that leverages Gaussian mixture models (GMMs) and Mahalanobis distance (MD) to detect stealthy and adaptive poisoning attacks, (ii) an optimized privacy-preserving aggregation scheme based on transformed additive ElGamal encryption with adaptive pruning and quantization for secure and efficient communication, and (iii) a dual-objective personalized learning strategy that improves model adaptation under non-IID data using logit-adjusted loss. Extensive experiments on the N-BaIoT dataset under both IID and non-IID settings, including scenarios with up to 50\% adversarial clients, demonstrate that SecureDyn-FL consistently outperforms state-of-the-art FL-based IDS defenses. It achieves up to 99.01\% detection accuracy, a 98.9\% F1-score, and significantly reduced attack success rates across diverse poisoning attacks, while maintaining strong privacy guarantees and computational efficiency for resource-constrained IoT devices.
\end{abstract}


\section{Introduction}

The rapid expansion of the Internet has driven large-scale connectivity, accelerating the deployment of Internet of Things (IoT) devices \cite{AI6G_ICTE2022, What5G_S2017, ZIKRIA2018699, MurtazaIoT, WiFiCommag2021, IOTO_ACCESS18}. By 2025, IoT devices are expected to exceed 55.7 billion globally \cite{hojlo2021future}. Smart appliances and sensors produce massive data streams but have limited computing power and minimal security \cite{blockchain_IoT2023}. These vulnerabilities make them attractive targets for cyber adversaries, as illustrated by the 2016 Mirai botnet attack \cite{antonakakis2017understanding}, which remains one of the most severe IoT security incidents. Such events highlight the need for effective IoT-specific cybersecurity measures. Intrusion Detection Systems (IDS), which analyze network traffic to identify threats, are essential components for IoT defense \cite{IoTEAAI_2024}.

Federated Learning (FL) offers a promising paradigm for addressing the privacy and scalability limitations of centralized IDS \cite{mcmahan2017communication}. In FL, clients collaboratively train a shared model without exchanging raw data. Clients compute local updates that are aggregated by a central server to form a global model. This iterative process preserves privacy, reduces communication overhead, and supports distributed intrusion detection \cite{agrawal2022federated}. Recent FL-based IDS frameworks \cite{friha2022felids, kelli2021ids} leverage this approach to enhance scalability and privacy.

Despite these advantages, several deployment challenges remain. IoT environments typically exhibit non-independent-and-identically-distributed (non-IID) data, with clients receiving heterogeneous mixes of benign and malicious traffic. Such heterogeneity degrades global convergence and detection performance \cite{agrawal2022federated, ferrag2021federated, mothukuri2021survey}. FL is also vulnerable to poisoning attacks: compromised clients can inject malicious updates to bias or degrade the global model \cite{fang2020local, chang2023privacy}. The lack of client supervision amplifies these threats \cite{rey2022federated, zhang2022secfednids}. Adversaries can further exploit eavesdropping or man-in-the-middle (MITM) attacks \cite{tt3, tt4} to reconstruct local models and conduct inference attacks \cite{xu2021else, driouich2022novel, wang2019eavesdrop}. Gradient manipulation techniques can bypass traditional defenses \cite{wan2021shielding}.

Privacy-preserving techniques such as differential privacy (DP), secure multiparty computation (SMC), and homomorphic encryption (HE) have been proposed to mitigate these threats \cite{xu2022hercules, zhao2019privacy}. DP-based methods \cite{chandu2025federated, li2023efficient} inject noise to protect data but often degrade accuracy in dynamic or imbalanced settings. SMC approaches \cite{liu2024survey} prevent data leakage but incur heavy computational and communication costs, unsuitable for constrained devices. HE frameworks \cite{guo2024efficient} ensure confidentiality but increase latency. Hybrid solutions \cite{sebert2022protecting, xu2022privacy, sav2022privacy} aim to balance privacy and performance but face trade-offs. Frameworks such as TrustFL \cite{zhang2020enabling} and SafeFL \cite{gehlhar2023safefl} advance the field but have practical limitations: TrustFL depends on trusted execution environments (TEEs), and SafeFL incurs high computational overhead. Many poisoning defenses require gradient access, which risks data leakage, and existing methods struggle against encrypted poisoning attacks.

These limitations can be summarized as: (i) vulnerability to poisoning and inference attacks, (ii) poor adaptation to non-IID data and limited generalization to unseen intrusions, and (iii) communication inefficiencies unsuitable for resource-constrained devices. While recent zero-shot learning (ZSL) advances \cite{asif2024advanced} offer promising techniques for generalization, further progress is needed to address these issues holistically.

This paper proposes SecureDyn-FL, a privacy-preserving FL-based IDS tailored for heterogeneous IoT environments. SecureDyn-FL addresses three core challenges: data heterogeneity, poisoning resilience, and communication security. A model decoupling strategy separates shared feature extraction from personalized classification, allowing local adaptation while maintaining global robustness. A joint mini-batch logit-adjusted and cross-entropy loss improves learning under diverse distributions. Unlike standard aggregation methods, which are vulnerable to label-flipping and model poisoning \cite{biggio2012poisoning, bagdasaryan2020backdoor}, SecureDyn-FL maintains stable convergence under heterogeneous and adversarial conditions.


SecureDyn-FL ensures secure communication by integrating homomorphic encryption with efficient privacy-preserving mechanisms. This design enables encrypted computation without compromising computational efficiency, thereby safeguarding both model parameters and sensitive data throughout the training process. The framework is engineered for real-world deployment, delivering strong security guarantees, robustness against adversarial behavior, and scalability across large-scale distributed environments.

The main contributions are:
\begin{itemize}
\item Proposing \textbf{SecureDyn-FL}, integrating dynamic auditing, encryption, and personalization to address heterogeneity, poisoning, and communication security in FL-based IDS.
\item Developing a Gaussian mixture model (GMM)-based auditing mechanism using Mahalanobis distance (MD) to detect stealthy and adaptive poisoning attacks.
\item Designing a hybrid loss function that improves local adaptation and global performance.
\item Applying dynamic pruning and quantization to reduce communication overhead without degrading accuracy or privacy.
\item Demonstrating strong cross-dataset generalization, maintaining high performance across N-BaIoT, and $TON_{IoT}$ benchmarks.
\item Achieving up to 99.01\% overall accuracy and a 0.9893 F1-score on N-BaIoT under same-model poisoning, while maintaining strong robustness across all attack settings, consistently outperforming state-of-the-art defenses such as FL Trust, Shield FL, and FL-Defender.
\end{itemize}

The remainder of this paper is structured as follows. Section~\ref{problem_formulation} presents the problem formulation and threat model. Section~\ref{preliminaries} discusses the necessary background and related concepts. Section~\ref{sec_workflow} introduces the overall SecureDyn-FL framework and system workflow, while Section~\ref{proposed_model} details the proposed model and defense mechanisms. Section~\ref{theoretical_analysis} provides theoretical and security analysis. Section~\ref{experiments} presents the experimental setup and results. Section~\ref{comparison} offers a comparative evaluation with existing state-of-the-art methods. Section~\ref{complexity_analysis} analyzes computational efficiency and scalability. Finally, Section~\ref{conclusion} concludes the paper and outlines future research directions.


\section{PROBLEM FORMULATION}
\label{problem_formulation}

In SecureDyn-FL models, a key focus is to secure the entire FL pipeline, with a particular emphasis on the secure aggregation of model updates. Clients contribute encrypted gradients to the FL server, bolstering data privacy by safeguarding individual contributions from both external adversaries and the server itself. Secure aggregation in SecureDyn-FL is distinct from simply sending encrypted gradients; it involves combining these gradients in a manner that prevents revealing individual updates while still facilitating effective global model training. 

Given $K$ clients with a dataset $D_i$ that may follow either IID or non-IID data distributions, the goal is to collaboratively train a global model $M$ by minimizing the loss function $L$. This must be achieved while ensuring data privacy, robustness against adversarial attacks, and minimal computational overhead. The SecureDyn-FL framework $F$ is designed to meet the following objectives:
\begin{itemize}
  \item Handle both IID and non-IID data distributions across clients.
   \item Ensure data privacy by allowing clients to contribute to $M$ without exposing their individual data.
    \item Verify encrypted gradients $g_i = \nabla (D_i , M)$ to detect malicious attacks.
     \item Implement secure aggregation of verified gradients to enhance model robustness against poisoning attacks.
      \item Maintain high performance of $M$ for all clients, irrespective of data distribution variations.
       \item Minimize the overhead of encryption and gradient verification while optimizing cryptographic efficiency.
       \end{itemize}
       
       The adversarial attack problem in SecureDyn-FL is formulated as minimizing $L(M)$ under privacy, robustness, and efficiency constraints, as expressed in Eq. ~\ref{eq1}.


\begin{equation}
\min_{M} L(M) = \left( \frac{1}{K} \right) \sum_{i=1}^{N} L(D_i , M)
\label{eq1}
\end{equation}

When optimizing Eq. ~\ref{eq1}, the following constraints are considered:
\begin{itemize}
  \item The global model $M$ is updated using verified gradients $g_i$ after their computation. 
    \item Gradients $g_i = \nabla(D_i , M)$ for $i = 1, \ldots , K$ are encrypted to preserve confidentiality.
      \item Performance constraints, such as data security, model robustness, adaptability to IID and non-IID data, and cryptographic efficiency, are satisfied.  
      \end{itemize}

The gradients are verified using a predefined verification
function $Va(g_i)$, as defined in Eq. ~\ref{eq2}.

\begin{equation}
Va(g_i) =
\begin{cases}
\text{benign} & \text{if } g_i \in \text{non-malicious} \\
\text{malicious} & \text{otherwise}
\end{cases}
\label{eq2}
\end{equation}

In Eq. ~\ref{eq2} $V_a(g_i)$ denotes a binary auditing decision variable produced by the temporal gradient auditing mechanism. A value of $V_a(g_i) = 1$ indicates that the client update $g_i$ is classified as benign and is therefore accepted for aggregation, whereas $V_a(g_i) = 0$ denotes that the update is identified as malicious and consequently rejected. This explicit interpretation improves the precision and interpretability of the proposed formulation. This function evaluates $g_i$ against expected non-malicious patterns, categorizing it accordingly. Such a mechanism ensures only valid gradients contribute to updating the global model $M$, thus enhancing its accuracy and resilience in adversarial SecureDyn-FL scenarios. Based on Eq. 2, the benign gradients are passed to the global model:

\begin{equation}
M = M - \eta \times \left( \frac{1}{K} \right) \times \sum_{i=1}^{K} Va(g_i) \times g_i
\end{equation}
where $\eta$ is the learning rate. The framework $F$ also evaluates accuracy, computation time, communication cost, and resource utilization to ensure practical and efficient real-world deployment.

\begin{figure}
    \centering
    \includegraphics[width=1\linewidth]{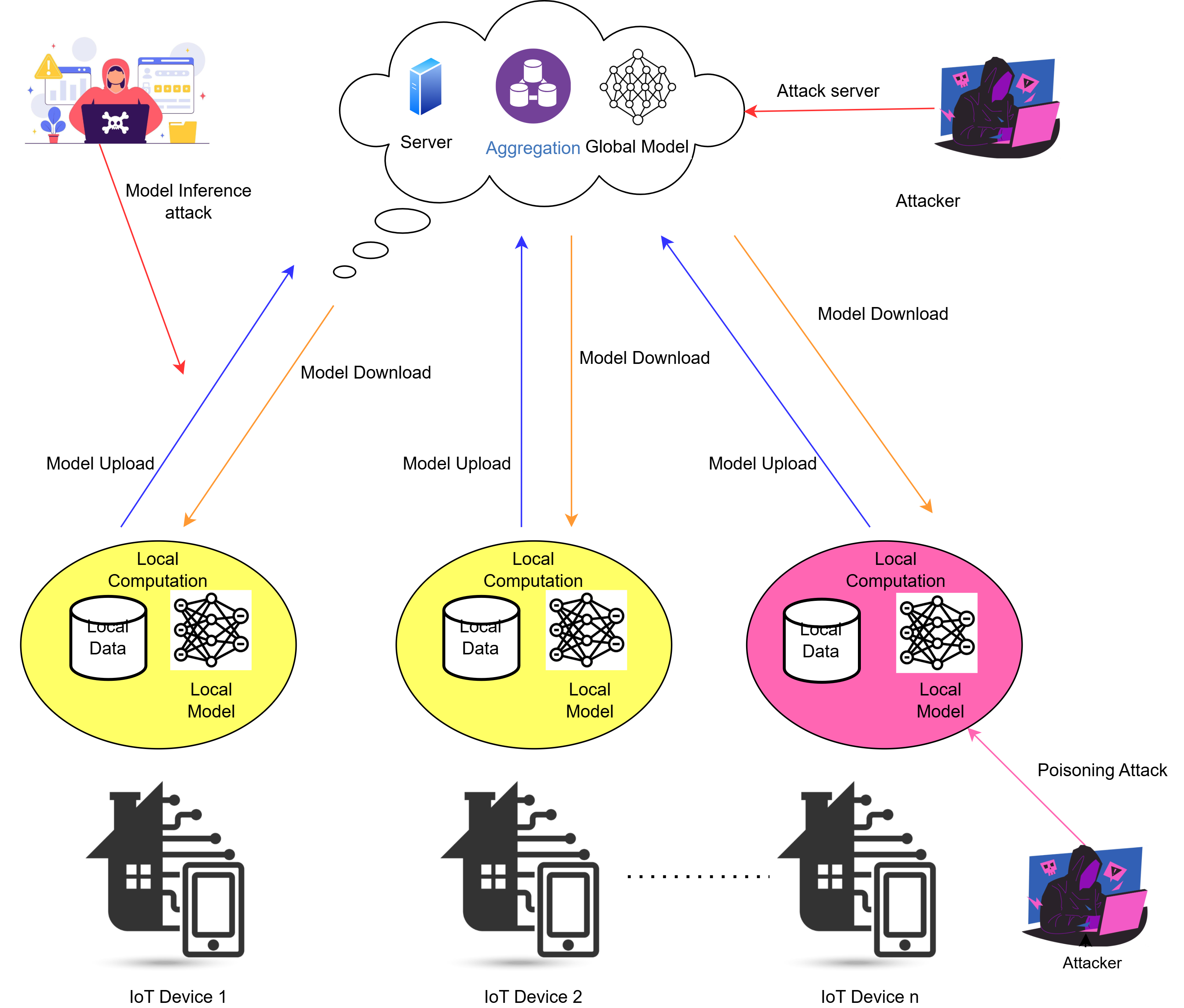}
    \caption{Thread Model: Federated Learning-based Intrusion Detection System (FL-IDS) in IoT networks with model inference and poisoning attacks.}
    \label{Thread}
\end{figure}

\subsection{Threat Model}
This paper focuses on security and privacy threats in FL that are exploitable by malicious users. Similar to \cite{xu2019verifynet}, \cite{liu2021privacy}, \cite{ma2022shieldfl}, and \cite{cao2020fltrust}, we classify users into two categories: benign or malicious, where malicious users (poisoners) may have access to diverse local datasets. Although advanced adaptive poisoning attacks exist, this study focuses on conventional model poisoning threats to maintain a consistent basis for comparison. We explicitly assume that users do not collude, excluding user-user collusion from this scope. Additionally, complex scenarios, such as server-client collusion or sophisticated adaptive attacks, are beyond the current focus. This limitation ensures a fair comparison with \cite{liu2021privacy}, \cite{ma2022shieldfl}, \cite{fang2020local}, \cite{blanchard2017machine}, and \cite{shen2016auror}, but also sets the stage for future research to extend our model to these more intricate threats, including various collusion scenarios and auditor threats. The threat model for our FL-based IDS in IoT networks, accounting for both model inference and poisoning attacks, is illustrated in Fig.~\ref{Thread}. In this setting, an honest-but-curious server may attempt to launch model inversion or related inference attacks by exploiting shared gradients to reconstruct sensitive local data, while malicious clients may perform targeted or untargeted poisoning to degrade detection performance. SecureDyn-FL mitigates these threats through additive homomorphic encryption, which prevents the server from accessing raw gradients, and temporal gradient auditing, which detects and filters abnormal client updates before aggregation.





The specific threats and corresponding goals are:

\begin{enumerate}
\item \textbf{Threat 1: The honest-but-curious} (HBC) adversary: In FL, the server has access to all local gradients and ciphertexts and may act adversarially. The system operates on the assumption of this HBC behavior, but the server could potentially launch privacy attacks, including inferring the data privacy of users. The core threat lies in adversaries seeking sensitive global model information through data reconstruction or inference attacks. 

\textbf{Goal 1:} Safeguard the confidentiality of local gradients. Adversaries, including malicious servers, can exploit shared gradients and global parameters to expose sensitive user data. Encrypting individual gradients before server transmission provides some degree of confidentiality.

\item \textbf{Threat 2: Inject poisonous gradients:} Byzantine actors can disrupt FL systems by submitting fraudulent gradients that mimic legitimate updates from heterogeneous data, compromising model integrity.

\textbf{Goal 2:} Enhance the examination of encrypted gradients to distinguish benign from malicious updates, strengthening resilience against poisoning attacks.
\end{enumerate}

\subsection{Design Goals}
SecureDyn-FL aims to achieve the following design goals to ensure high accuracy, robustness, privacy, and efficiency, even under adversarial conditions:

\begin{enumerate}
\item \textbf{Accuracy:} Maintain high classification accuracy across all clients, despite data imbalance, distribution skew, or adversarial attacks. A dual-objective loss function ensures dependable performance in both IID and non-IID settings.

\item \textbf{Robustness:} Ensure resilience against targeted and untargeted poisoning attacks, including adversaries who change their strategies over time. This is achieved through a dynamic temporal gradient auditing mechanism that tracks update behavior using GMM clustering and MD.

\item \textbf{Privacy:}
Sensitive data is protected during training and communication using additive homomorphic encryption based on a modified ElGamal scheme, enabling secure gradient aggregation without exposing raw updates.

\item \textbf{Adaptability:} Support client-level personalization via logit-adjusted loss, allowing local models to adapt to diverse class distributions while contributing to global learning, vital for heterogeneous IoT environments.

\item \textbf{Efficiency in Communication and Computation:}
Incorporate adaptive quantization and dynamic unstructured pruning to reduce overhead in low-power, bandwidth-constrained devices while preserving model quality.
\end{enumerate}

\section{Literature Review}
\label{preliminaries}
The proliferation of IoT devices has generated massive amounts of distributed data, necessitating learning frameworks that preserve privacy while enabling effective model training. FL has emerged as a promising decentralized paradigm, allowing multiple edge devices to collaboratively train a global model by exchanging only model parameters rather than raw data. This approach significantly mitigates privacy risks and communication overhead compared to centralized machine learning models. Despite these advantages, FL remains highly vulnerable to model poisoning and backdoor attacks, wherein malicious clients manipulate local gradients to degrade or subvert the global model. Fang et al. \cite{fang2020local} conducted the first systematic study of local model poisoning against Byzantine-robust FL methods. Their findings revealed that existing defenses, which were assumed to be robust against Byzantine failures, can be substantially compromised across multiple real-world datasets, thereby exposing critical vulnerabilities in federated optimization.

In response, researchers have proposed several defense strategies to enhance the robustness of FL. Jebreel and Domingo-Ferrer \cite{jebreel2023fl} introduced FL-Defender, which identifies attack-related neurons by analyzing last-layer gradient behaviors and re-weights client updates based on worker-wise angle similarity with PCA compression. Similarly, Gill et al. \cite{gill2023feddefender}  proposed FedDefender, leveraging differential testing on synthetic inputs and neuron activation fingerprinting to identify backdoor-infected clients, achieving attack success rates as low as 10\%. Erbil and Gursoy \cite{erbil2022defending}  explored a clustering-based defense using X-Means to isolate malicious updates by selectively extracting indicative DNN parameters, yielding up to 95\% true positive rates. To address large-scale attacks, Zhang et al. \cite{zhang2022fldetector} developed FLDetector, which detects malicious clients by analyzing model-update inconsistencies across iterations. Malicious participants are identified by their persistent deviation from expected update patterns, allowing Byzantine-robust methods to effectively train accurate models after removing compromised clients. Collectively, these methods reflect a growing body of work aimed at enhancing the security of federated learning in distributed IoT environments.

IDSs play a crucial role in safeguarding IoT and industrial IoT (IIoT) infrastructures against cyber threats. Traditional centralized IDS approaches face significant challenges related to data privacy, communication overhead, and real-time processing in distributed networks. FL-IDS has therefore emerged as an effective solution, combining collaborative model training with privacy preservation. Bhavsar et al. \cite{bhavsar2024fl} developed an FL-IDS using logistic regression and CNN classifiers, achieving 94–99\% accuracy on NSL-KDD and Car-Hacking datasets when deployed on low-power embedded devices such as Raspberry Pi. Akinie et al. \cite{akinie2025fine} proposed a hybrid server–edge framework that reduced memory consumption by 42\% and training time by 75\%, while maintaining 99.2\% detection accuracy. Javeed et al. \cite{javeed2024federated} combined CNN and BiLSTM architectures within a zero-trust FL model to capture spatial–temporal features, demonstrating strong performance on CICIDS2017 and Edge-IIoTset datasets. Rashid et al. \cite{rashid2023federated} further achieved 92.49\% accuracy on Edge-IIoTset, closely approaching the performance of centralized machine learning (93.92\%), thus validating the effectiveness of FL for intrusion detection without compromising privacy.

In IIoT environments, Ruzafa-Alcazar et al. \cite{ruzafa2021intrusion} conducted a comprehensive evaluation of differential privacy techniques applied to FL-IDS, comparing FedAvg and Fed+ aggregation methods on the $TON_{IoT}$ dataset under non-IID data distributions. Sun et al. \cite{hamdi2023federated} advanced this line of work by combining CNNs and Gated Recurrent Units with Isolation Trees for anomaly detection, improving both real-time performance and accuracy in non-IID settings. Similarly, Azeez et al. \cite{azeez2024federated} applied federated averaging on CICIDS2017, achieving 95.2\% accuracy, while Mahmud et al. \cite{mahmud2024privacy} demonstrated over 90\% accuracy across multiple attack types, including DoS, DDoS, and ransomware, in IoT networks. Zakaria Abou El Houda et al. \cite{abou2023secure} proposed an innovative approach that integrates secure aggregation protocols with blockchain technology to ensure both data integrity and privacy. Their system employs multi-party computation to prevent data exposure between participants and uses blockchain to provide tamper-resistance, achieving high detection accuracy on real-world IoT datasets. Collectively, these studies demonstrate that FL-IDS offers a scalable, privacy-preserving, and accurate intrusion detection mechanism for resource-constrained IoT environments.

Although FL inherently improves data privacy by avoiding raw data transmission, model updates themselves can leak sensitive information. Moreover, defense mechanisms against poisoning attacks must be designed carefully to avoid privacy compromises. As a result, recent research has focused on integrating advanced cryptographic techniques into federated learning to achieve both security and privacy. Yazdinejad et al. \cite{yazdinejad2024robust} proposed an internal auditing mechanism utilizing GMM and MD with additive homomorphic encryption (AHE) to detect malicious encrypted gradients while minimizing computational overhead. Miao et al. \cite{miao2024rfed} introduced RFed, a dual-server framework that employs scaled dot-product attention to achieve over 96\% poisoning attack failure rates without relying on strong assumptions, thereby improving scalability and robustness. Wu et al. \cite{wu2025privacy} developed PBFL, integrating two-trapdoor fully homomorphic encryption with secure normalization and cosine similarity methods to defend against poisoning while preventing privacy leakage during detection. Similarly, Ma et al. \cite{ma2022shieldfl} proposed ShieldFL, which uses two-trapdoor homomorphic encryption and secure cosine similarity measurements. ShieldFL achieved 30\%–80\% accuracy improvements against state-of-the-art poisoning attacks under both IID and non-IID data conditions. These advances illustrate a growing emphasis on privacy-preserving federated defense mechanisms, combining cryptographic primitives with statistical detection techniques to enable secure and practical FL deployment in large-scale IoT and IIoT networks.

Privacy preservation in federated learning has been extensively explored through cryptographic and statistical techniques, including secure multiparty computation (SMC), fully homomorphic encryption (FHE), differential privacy (DP), and more recently, blockchain-based mechanisms. SMC and FHE provide strong security guarantees by enabling computations on encrypted data without exposing individual client updates. However, these approaches typically incur significant computational and communication overhead, limiting their scalability in resource-constrained IoT environments \cite{xu2022hercules, zhao2019privacy}. DP techniques introduce calibrated noise to model updates, effectively protecting individual data privacy but often at the cost of reduced model utility and convergence performance, particularly in non-IID settings \cite{sebert2022protecting, sav2022privacy}. To balance privacy and efficiency, hybrid approaches have emerged that combine DP with SMC or homomorphic encryption, aiming to reduce overhead while maintaining strong privacy guarantees. In parallel, blockchain-based frameworks have been proposed to improve the integrity and auditability of the aggregation process. By leveraging immutable ledgers and decentralized consensus, blockchain can prevent tampering and ensure trustworthy model updates without relying on a fully trusted central server \cite{sav2020poseidon, gehlhar2023safefl, salam2023efficient}. While these methods offer promising privacy protection, most existing solutions struggle to simultaneously ensure lightweight operation, robust poisoning defense, and high detection accuracy in heterogeneous IoT deployments.

Existing federated learning-based IDS frameworks face notable limitations in jointly addressing data heterogeneity, adaptive poisoning resilience, and communication privacy in realistic IoT environments. Current defenses often rely on static detection mechanisms, trusted hardware, or heavy cryptographic schemes, which either fail against stealthy attacks or introduce prohibitive overhead in resource-constrained settings. Moreover, non-IID data distributions significantly hinder global model convergence and detection accuracy, a challenge insufficiently addressed by prior works. To fill these gaps, this paper proposes SecureDyn-FL, a federated intrusion detection framework that integrates dynamic temporal gradient auditing, lightweight privacy-preserving encryption, and personalized learning strategies. This holistic approach simultaneously strengthens resilience against adaptive poisoning, enhances privacy protection, and improves detection performance under non-IID conditions, thereby advancing the state of the art in secure and efficient FL-based intrusion detection for IoT networks.

\section{WorkFlow}\label{sec_workflow}

The SecureDyn-FL framework follows a structured sequence of coordinated phases that together form an end-to-end pipeline for secure, personalized, and robust intrusion detection in federated IoT environments. Rather than treating each component in isolation, the workflow integrates personalized learning, communication efficiency, and security mechanisms into a cohesive process. Client registration establishes trust and tracking, personalized training ensures adaptability to non-IID data, pruning and quantization reduce communication cost, encryption protects updates during transmission, and temporal gradient auditing detects poisoning attacks before aggregation. This high-level structure illustrates how each stage contributes directly to the overarching IDS goal: detecting malicious behavior while preserving data privacy and efficiency in federated settings.

The proposed SecureDyn-FL framework operates through a sequence of coordinated phases to ensure secure, efficient, and personalized FL for intrusion detection in IoT environments as shown in Fig.~\ref {fig_workflow}. The system initializes the federated infrastructure by assigning unique identifiers and cryptographic key pairs to each participating client. Simultaneously, a Central Audit (CA) module prepares a tagging and update tracking mechanism for later-stage gradient verification and poisoning detection.

Each client performs personalized local training using its non-IID data. The local model is decoupled into a shared feature extractor and a private classifier. A dual-loss strategy, comprising cross-entropy and mini-batch logit adjustment, is used to enhance robustness against class imbalance and heterogeneous distributions. Following local training, clients apply soft, unstructured L1-norm-based pruning with a dynamically increasing pruning rate over rounds to reduce computational and communication overhead while improving privacy. Subsequently, clients quantize the pruned updates using an adaptive quantization method. This involves computing a scale and zero-point to map real-valued updates into a lower bit-width representation, reducing bandwidth usage without sacrificing accuracy. These quantized updates are then encrypted using the CKKS homomorphic encryption scheme, enabling secure aggregation at the server without decryption.

Before aggregation, a CA is conducted. The first phase validates client identities and updates tags. The second phase involves clustering using incremental GMMs, evaluating MD, and verifying temporal trajectory consistency of updates. A dynamic multi-threshold decision mechanism classifies client updates as accepted, down-weighted, or rejected. Only verified and reliable gradients proceed to the aggregation step. The server aggregates the validated updates and constructs a global model, which is redistributed to clients for the next training round. Clients integrate the updated model while continuing to refine local classifiers. This workflow collectively ensures secure communication, personalized learning, defense against poisoning attacks, and robustness in heterogeneous IoT environments.

\begin{figure*}
    \centering
    \includegraphics[width=6.5in,height=3.6in]
{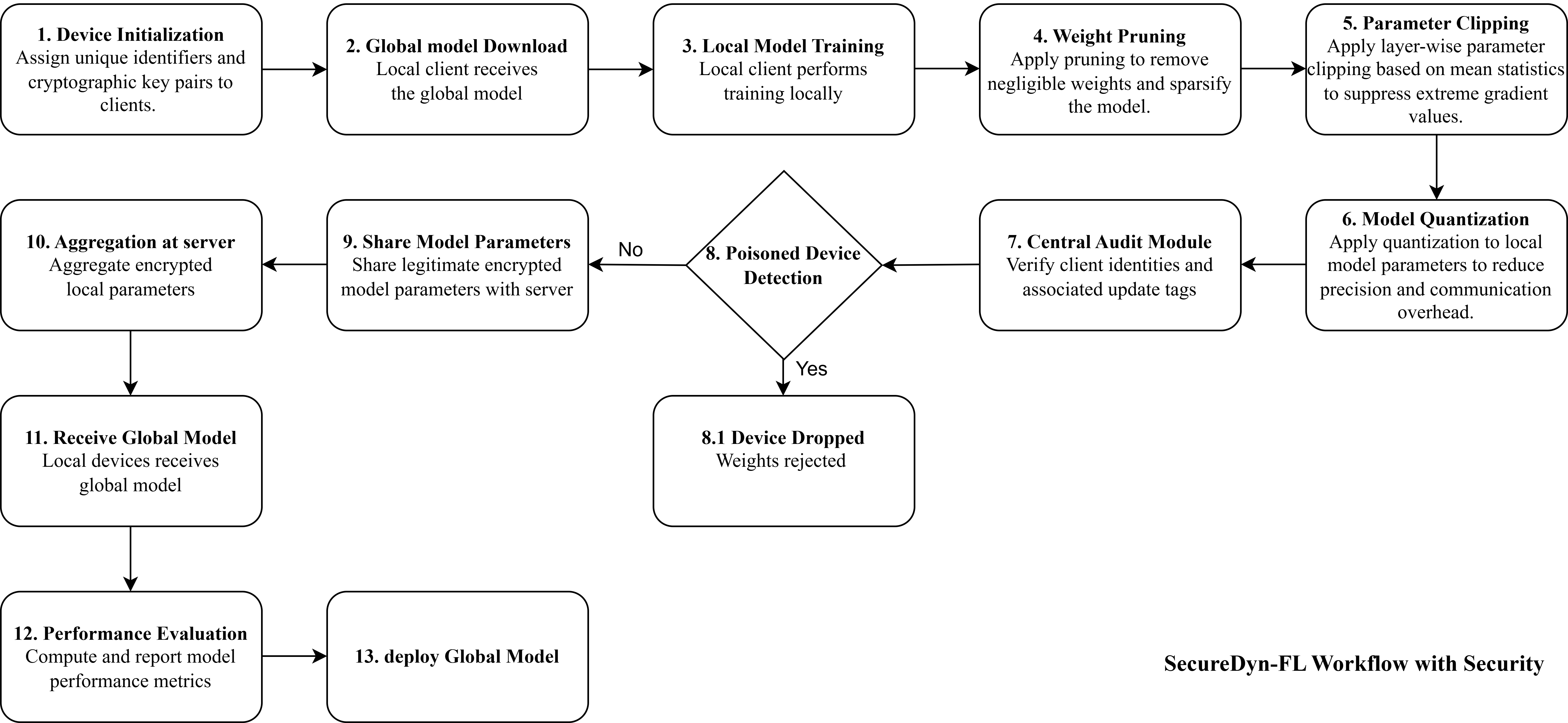}
    \caption{Flowchart of the SecureDyn-FL framework, including device initialization, local model training on clients, security measures (e.g., poisoned device detection and central auditing), optimization (via weight pruning and quantization), privacy-preserving aggregation at the server, and performance evaluation before deploying the global model.}
    \label{fig_workflow}
\end{figure*}

\section{Proposed Model}
\label{proposed_model}

This section presents proposed SecureDyn-FL, a novel personalized FL framework tailored for intrusion detection in non-IID IoT environments. To address the vulnerabilities of FL against poisoning attacks, we introduce a server-side poisoned client detection mechanism before the global model aggregation step of FL.

\subsection{System Architecture Overview}
The architecture consists of four key entities:
\begin{itemize}
    \item Clients (IoT Devices): Distributed nodes that locally train on non-IID data while preserving data privacy.
    \item Central Auditor: A trusted party responsible for client registration, key distribution, secure aggregation, model auditing, and dissemination.
    \item Communication Network: Secure channels through which encrypted updates and models are exchanged.
    \item Audit Table Repository: A secure, encrypted ledger storing client-specific encrypted updates and behavioral metadata for auditing.
\end{itemize}

The architecture supports secure model updates, client personalization, encryption-enabled communication, and integrity verification, without revealing sensitive client information or raw data.

\begin{figure*}
    \centering
        \includegraphics[width=7in,height=6.5in]{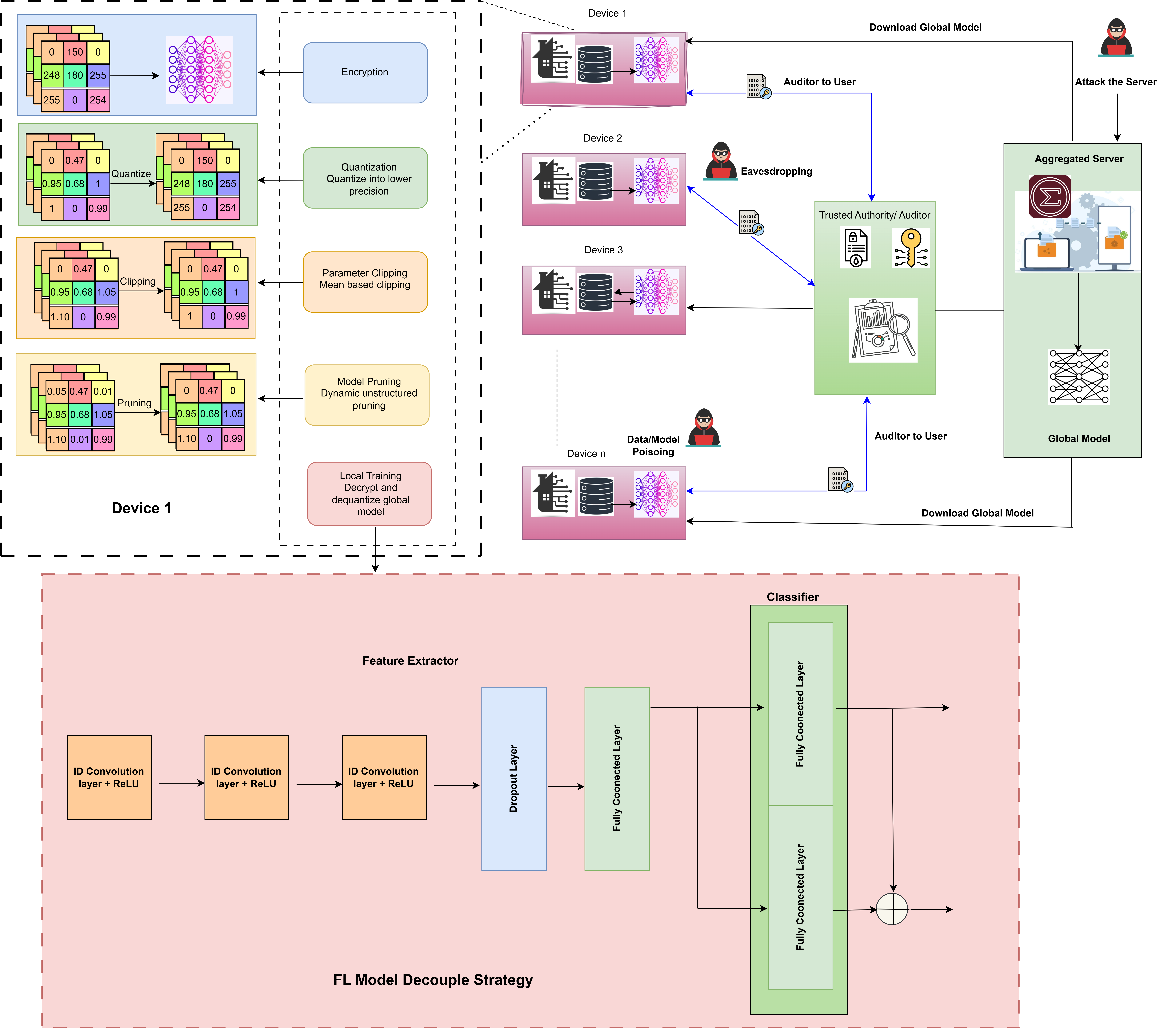}
    \caption{System model of proposed SecureDyn-FL framework. The architecture integrates multiple defense mechanisms, including quantization, mean-based parameter clipping, and dynamic unstructured model pruning, followed by encrypted local training. The FL model is decoupled into a shared feature extractor and dual classifiers (global and personalized) enabling multi-objective optimization. A trusted auditor monitors the training pipeline to detect adversarial threats such as data/model poisoning and eavesdropping. Secure aggregation ensures the integrity and privacy of the global model.}
    \label{arch}
\end{figure*}

\subsection{Registration Phase}
\label{R.P}

The device registration and authentication process comprises four main phases: registration, login, key exchange, and authentication. During registration, each device is enrolled with the certification authority (CA), which assigns a unique tag identity (TID) and generates a secret value (HSV) based on the device’s MAC address. A hashed identity (HID) is derived and stored securely to prevent duplicate registrations \cite{ashraf2023lightweight, soomro2024lightweight}. In the login phase, the device submits its HID for server validation, followed by the exchange of masked random values to reconstruct shared secrets. In the key exchange phase, both parties derive a 128-bit session key through random number operations and bitwise transformations, ensuring confidentiality. Finally, mutual authentication is achieved using HMAC-SHA256 over session keys, IP addresses, and nonces, securing the communication channel against MITM attacks. For a detailed algorithm,  mathematical derivations, and message exchanges, the reader is referred to \cite{ashraf2023robust}.

\subsection{Proposed Method: Multi-Objective Personalized Federated Learning}

To tackle the challenges introduced by statistical heterogeneity (i.e., non-IID data distributions) in FL, we propose a multi-objective personalized FL (MO-PFL) framework. This approach adopts a decoupled dual-classifier strategy coupled with a tailored multi-objective loss function. Our architecture enables client-level personalization while preserving alignment with the global learning objective. As shown in Fig.~\ref{arch}, the model consists of a \textit{shared feature extractor} and two decoupled classifiers: a \textit{personalized classifier} tuned to the local distribution and a \textit{global classifier} synchronized across clients.

\subsection{Model Architecture}

Consider a client $k \in \{1, 2, \dots, K\}$ with local dataset $D_k = \{(\mathbf{x}_i, y_i)\}_{i=1}^{n_k}$ sampled from a possibly non-IID distribution $\mathcal{D}_k$. The local model comprises:

\begin{itemize}
    \item A shared feature extractor: $f: \mathbb{R}^d \rightarrow \mathbb{R}^p$
    \item A personalized classifier: $h^{\text{pers}}_k: \mathbb{R}^p \rightarrow \mathbb{R}^{|C|}$
    \item A global classifier: $h^{\text{glob}}_k: \mathbb{R}^p \rightarrow \mathbb{R}^{|C|}$
\end{itemize}

For input $\mathbf{x}$ on client $k$, the model produces:

\begin{equation}
\hat{y}^{\text{pers}} = h^{\text{pers}}_k(f(\mathbf{x})), \quad \hat{y}^{\text{glob}} = h^{\text{glob}}_k(f(\mathbf{x}))
\end{equation}

The extractor $f$ is globally shared and updated via the \texttt{FedAvg} algorithm~\cite{mcmahan2017communication}, while $h^{\text{pers}}_k$ remains private to each client for local adaptation. The global classifier $h^{\text{glob}}_k$ is periodically synchronized across clients.

\subsection{Multi-Objective Loss Formulation}

To balance personalization and global generalization, we define a composite loss function:

\paragraph{Global Cross-Entropy Loss}
This loss ensures consistency across clients by training the global classifier:

\begin{equation}
\mathcal{L}_{\text{CE}}(y, \hat{y}^{\text{glob}}) = -\log \left( \frac{\exp(\hat{y}^{\text{glob}}_y)}{\sum_{y' \in C} \exp(\hat{y}^{\text{glob}}_{y'})} \right)
\end{equation}
where $C$ is the label space and $\hat{y}^{\text{glob}}_{y'}$ is the output logits for class $y'$. $y^{\text{glob}}$ as the predicted label generated by the global classifier on each client. This value is utilized in the computation of the global cross-entropy loss to ensure alignment between local and global objectives during federated training.
\paragraph{Logit-Adjusted Personalization Loss}
To handle the local objective of the client and learn the personalized model for each client, we apply a logit-adjusted cross-entropy loss to the personalized classifier:

\begin{equation}
\mathcal{L}_{\text{LA}}(y, \hat{y}^{\text{pers}}) = -\log \left( \frac{\exp(\hat{y}^{\text{pers}}_y + \tau \log \alpha_y^k)}{\sum_{y' \in C} \exp(\hat{y}^{\text{pers}}_{y'} + \tau \log \alpha_{y'}^k)} \right)
\end{equation}
where $\alpha_y^k$ is the normalized frequency of class $y$ in client $k$. It is dynamically recalculated at the mini-batch level, enabling real-time adaptation to batch-wise label skew. The temperature parameter $\tau$ regulates the intensity of the logit adjustment.

\paragraph{Mitigating Non-IID Challenges.}
Beyond addressing class imbalance, the proposed dual-objective formulation plays a crucial role in mitigating the adverse effects of \textit{non-IID data distributions} commonly observed in IoT environments. The global cross-entropy loss $\mathcal{L}_{\text{CE}}$ enforces alignment with the global decision boundary by optimizing shared representations collaboratively across clients, ensuring that local models do not drift too far from the global objective. Meanwhile, the logit-adjusted personalization loss $\mathcal{L}_{\text{LA}}$ dynamically reweights logits based on local class priors, which counteracts the bias introduced by label distribution skew and stabilizes local training. By jointly optimizing these objectives, the model achieves a balance between \textit{local specialization} and \textit{global consistency}, reducing client-drift and improving the robustness of aggregation under heterogeneous non-IID settings.

\subsection{Optimization and Gradient Flow}

The total loss for each client is:

\begin{equation}
\mathcal{L}_{\text{total}} = \underbrace{\mathcal{L}_{\text{CE}}(y, \hat{y}^{\text{glob}})}_{\text{Global objective}} + \lambda \underbrace{\mathcal{L}_{\text{LA}}(y, \hat{y}^{\text{pers}})}_{\text{Personalization objective}}
\end{equation}

\paragraph{Backward Pass}
\begin{itemize}
    \item $h^{\text{pers}}_k$ is updated using $\nabla \mathcal{L}_{\text{LA}}$
    \item $h^{\text{glob}}_k$ is updated using $\nabla \mathcal{L}_{\text{CE}}$
    \item $f$ is updated using $\nabla_f (\mathcal{L}_{\text{CE}} + \mathcal{L}_{\text{LA}})$
\end{itemize}

\subsection{Federated Training Protocol}

Each local training round on client $k$ proceeds through the following steps:

\begin{enumerate}
    \item \textbf{Forward Pass:} \\
    The client computes intermediate features using the shared feature extractor: 
\begin{equation}
    \mathbf{z} = f(\mathbf{x})
    \end{equation}
    These features are passed to both classifiers:
 \begin{equation}
    \hat{y}^{\text{pers}} = h^{\text{pers}}_k(\mathbf{z}), \quad \hat{y}^{\text{glob}} = h^{\text{glob}}_k(\mathbf{z})
 \end{equation}
    Here, $\hat{y}^{\text{pers}}$ reflects the personalized inference tailored to local data, while $\hat{y}^{\text{glob}}$ contributes to the globally shared model.

    \item \textbf{Loss Computation:} \\
    The total loss is formulated as a weighted sum of the global cross-entropy loss and the logit-adjusted personalization loss:
 \begin{equation}
    \mathcal{L}_{\text{total}} = \mathcal{L}_{\text{CE}}(y, \hat{y}^{\text{glob}}) + \lambda \mathcal{L}_{\text{LA}}(y, \hat{y}^{\text{pers}})
  \end{equation}
    The parameter $\lambda$ balances the trade-off between global consistency and local adaptation.

    \item \textbf{Backpropagation:} \\
    Gradients are computed and applied separately:
     \begin{align}
    h^{\text{pers}}_k &\leftarrow h^{\text{pers}}_k - \eta \nabla_{h^{\text{pers}}_k} \mathcal{L}_{\text{LA}} \nonumber \\
    h^{\text{glob}}_k &\leftarrow h^{\text{glob}}_k - \eta \nabla_{h^{\text{glob}}_k} \mathcal{L}_{\text{CE}} \nonumber \\
    f &\leftarrow f - \eta \nabla_f \left( \mathcal{L}_{\text{CE}} + \mathcal{L}_{\text{LA}} \right)
\end{align}

    The feature extractor $f$ receives a combined gradient from both loss terms, enabling it to capture features that support both global generalization and local personalization.
\end{enumerate}

\subsection{Quantization and Pruning}

During local model training, we employed a pruning technique to iteratively remove less important weights or gradients from model updates. Specifically, clients perform soft unstructured pruning based on the L1 norm, which creates a sparse model and makes the FL training process more efficient. The pruning process is guided by a dynamically updated pruning rate $p_t$, which increases over the communication rounds, allowing more aggressive pruning as the training progresses. After pruning, clients send their pruned updates to the server, which aggregates them using FedAvg to generate the global model. This pruning technique not only reduces the model size and computational costs, but also makes the training process more resistant to inference attacks.

By progressively increasing the pruning rate, communication efficiency improves throughout the rounds. As clients share a sparsified model with the server, the transmitted model is no longer the full model, limiting the information available to potential attackers. The sparsity introduced by pruning constrains the parameter space, significantly reducing the chances of reverse engineering or inferring sensitive data. This reduction in exposed parameters inherently enhances privacy protection, making it more difficult for adversaries to extract meaningful insights about the underlying data.

The pruning rate $p_t$ is updated iteratively using Eq.~\ref{eq:pruning_rate}:
\begin{equation}
p_t = \max\left(0, \frac{t - t_{\text{eff}}}{t_{\text{target}} - t_{\text{eff}}}\right) \cdot (p_{\text{target}} - p_0) + p_0
\label{eq:pruning_rate}
\end{equation}
where $p_t$ is the pruning rate at round $t$, $t_{\text{eff}}$ is the effective round when pruning starts, $t_{\text{target}}$ is the target round when the target pruning rate is reached, $p_0$ is the initial pruning rate, and $p_{\text{target}}$ is the target pruning rate. This pruning rate increases gradually from the initial value to the target value, ensuring that the pruning is progressively applied more aggressively as training progresses.

After applying pruning to the model updates at each client, the pruned local update $\Delta w^{t+1}_{p,i}$ is computed as:
\begin{equation}
\Delta w^{t+1}_{p,i} = \Delta w^{t+1}_{i} \odot m^t_i
\label{eq:pruned_update}
\end{equation}
where $\odot$ represents the element-wise product, and $m^t_i$ is the local pruning mask identifying pruned weights at communication round  $t$. This pruned update $\Delta w^{t+1}_{p,i}$ is then quantized and sent to the server for aggregation.

We also used a layer-wise clipping technique based on dynamic mean to help reduce inconsistencies during the training process. The clipping factor controls the clipping parameter, dynamically adjusting the clipping based on layer-wise updates, rather than using a static clipping method. This approach ensures that each layer’s updates are clipped according to their specific dynamics, leading to more stable and efficient training. After the local model updates are computed, each client clips its own model update $\Delta w^{t+1}_i$ to avoid instability before sending it to the server. The clipping for client $i$’s model update is applied using Eq.~\ref{eq:clipping}:
\begin{equation}
\Delta w^{t+1}_{C,i} = \text{clip}(\Delta w^{t+1}_i, -\alpha \cdot \mu_i, \alpha \cdot \mu_i)
\label{eq:clipping}
\end{equation}
where $\mu_i$ is the mean of the absolute values of the elements of client $i$’s model update, calculated as:
\begin{equation}
\mu_i = \frac{1}{n} \sum_{j=1}^{n} \left| \Delta w^{t+1}_{i,j} \right|
\label{eq:mean_abs}
\end{equation}

The clipping function $\text{clip}(\Delta w^{t+1}_i, a, b)$ ensures that the values of $\Delta w^{t+1}_i$ are constrained within the range $[-\alpha \cdot \mu_i, \alpha \cdot \mu_i]$, thereby limiting the impact of extreme values.

Next, each client performs adaptive quantization (AQ), a novel technique designed to address the communication overhead challenges in FL for resource-constrained IoT environments. AQ scheme allows devices to dynamically adjust their quantization levels based on their available communication resources. Unlike traditional quantization methods that enforce a uniform strategy across all devices, AQ accounts for device heterogeneity, enabling efficient on-device training while maintaining model accuracy.

The AQ scheme employs a \(K\)-level quantizer, where \(K\) represents the number of distinct values to which each weight can be mapped. For each device \(i \in \mathcal{N}\), the model updates \(x_i\) are quantized to reduce their size before transmission. We assume that the elements of \(x_i\) fall within the range \([r_1, r_2]\).

The quantization process is defined as follows:
   - The range \([r_1, r_2]\) is partitioned into \(N\) contiguous intervals with equal probability.
   - The threshold value separating the \(l\)-th and \((l+1)\)-th intervals is given by Eq. \ref{equ1h}.
     
     \begin{equation}
     \label{equ1h}
         T(l) = r_1 + l \cdot \Delta_{N}
     \end{equation}
    where \(\Delta_{N}\) is the quantization interval, calculated using Eq. \ref{equ1i}. 
     \begin{equation}
     \label{equ1i}
         \Delta_{N} = \frac{r_2 - r_1}{N - 1}
     \end{equation}
    
Each weight \(x_i(k)\) is mapped to a discrete value \(\bar{x}_i(k)\) using the following stochastic quantization rule:
     \[
     Q_{N}(x_i(N)) =
     \begin{cases}  
       T(l + 1) & \text{with probability} \frac{x_i(N) - T(l)}{T(l+1) - T(l)}, \\
       T(l) & \text{otherwise}.
     \end{cases}
     \]
     
This rule assigns \(x_i(N)\) to the quantization level \(T(l+1)\) with a probability proportional to its distance from \(T(l)\), ensuring minimal distortion.
The quantized output \(\bar{x}_i(N)\) takes a discrete value from the set as shown in Eq. \ref{equ1j}.
    
    \begin{equation}
    \label{equ1j}
        \{ r_1, r_1 + \Delta_{N}, r_1 + 2\Delta_{N}, \ldots, r_2 - \Delta_{N}, r_2 \}
    \end{equation}
    
The key novelty of AQ lies in its ability to dynamically adjust quantization levels based on the communication resources of each device. Devices with limited bandwidth can use fewer quantization levels (smaller \(N\)), while devices with higher bandwidth can use more levels (larger \(N\)). This adaptive approach ensures that the communication overhead is minimized without compromising model accuracy. AQ is scalable to large IoT networks with diverse device capabilities, making it suitable for real-world deployments.




\subsection{Additive ElGamal encryption based on discrete logarithm}

After completing the quantization process, each client encrypts the quantized model updates using the PRKG algorithm. However, to facilitate secure aggregation in the FL environment, it is essential to endow the ElGamal cryptosystem with additive homomorphic properties. By default, the ElGamal cryptosystem supports multiplicative homomorphism, that is,
\begin{equation}
\text{Enc}(m_1) \cdot \text{Enc}(m_2) = \text{Enc}(m_1 \cdot m_1),
\end{equation}
but it does not naturally support additive operations. To achieve additive homomorphism necessary for FL model aggregation, Zhu et al.~\cite{zhu2021distributed} introduced a secure transformation technique based on the Cramer transformation, enabling additive homomorphism over ElGamal-encrypted data.

The transformation works by converting the original plaintext \( m \) to an exponentiated form \( m' = g^m \mod p \), thereby mapping the message space into \( \mathbb{Z}_p \). The transformed ElGamal scheme then facilitates additive homomorphism as demonstrated in Eq.~\eqref{eq:add_hom}:

\begin{align}
\text{Enc}(m_1) \cdot \text{Enc}(m_2) &= g^{m_1} y^{r_1} \cdot g^{m_2} y^{r_2} \notag \\
&= g^{m_1 + m_2} y^{r_1 + r_2} \mod p
\label{eq:add_hom}
\end{align}

Although this transformation enables additive homomorphism, it significantly increases the computational burden of encryption and decryption, primarily due to the necessity of discrete logarithm recovery. Zhu et al. employed brute-force search and logarithmic recovery techniques to retrieve plaintexts, which are computationally intensive.

To ensure privacy during the federated learning process, SecureDyn-FL employs a post-Cramer transformed additive ElGamal encryption scheme, enabling additive homomorphism while maintaining computational efficiency. This allows encrypted local model updates to be aggregated on the server without decryption, thereby preventing information leakage to both external adversaries and honest-but-curious servers. Unlike differential privacy, which introduces noise and can negatively impact utility, this encryption mechanism preserves the exact model updates, thus avoiding accuracy degradation during aggregation. The approach balances privacy and performance, offering a practical alternative to DP or fully homomorphic encryption, which often impose significant computational costs. The process involves the following steps:

First, select two large, distinct prime numbers \( p_p \) and \( q_p \) of equal bit-length, and compute \( n = p_p \cdot q_p \) and \( \lambda = \text{lcm}(p_p - 1, q_p - 1) \), ensuring that \( n^2 < p \). Then, choose a generator \( g_p \in \mathbb{Z}^*_{n^2} \) such that
\begin{equation}
\gcd(L(g_p^\lambda \mod n^2), n) = 1,
\end{equation}
where the function \( L(x) = \frac{x - 1}{n} \) denotes the L-function used in the Paillier cryptosystem. With this setup, the message \( m \) is transformed into \( m' = g_p^m \mod n^2 \), and the encryption scheme is expressed as:

\begin{equation}
c := \langle c_1 = g^r \mod p,\quad
c_2 = g_p^m \mod n^2 \cdot y^r \mod p \rangle
\label{eq:post_cramer_enc}
\end{equation}

This modified ciphertext structure retains additive homomorphism, enabling secure aggregation in FL. Upon decryption using ElGamal, the transformed message \( m' = g_p^m \mod n^2 \) is obtained. The original plaintext \( m \) is recovered through the expression:

\begin{equation}
m = \frac{L(g_p^m \mod n^2)}{L(g_p^\lambda \mod n^2)} \mod n
\label{eq:post_cramer_dec}
\end{equation}

Therefore, if two clients \( u_1 \) and \( u_2 \) share the public parameters \( (g_p, n, \lambda) \), they can independently decrypt the ciphertext and retrieve the plaintext \( m \). The complete encryption, homomorphic addition, and decryption process is summarized in Algorithm~\ref{alg:enc_dec}.

\begin{algorithm}[H]
\caption{Post-Cramer Transformation Based Encryption and Decryption}
\label{alg:enc_dec}

\textbf{Encryption:} $\text{Enc}_{pk}(m) \rightarrow \llbracket m \rrbracket$
\begin{algorithmic}[1]
\State Convert plaintext \( m \) to \( m' = g_p^m \bmod n^2 \) using post-Cramer transformation.
\State Compute ciphertext: 
\Statex \hspace{\algorithmicindent} \( c_1 = g^r \bmod p \)
\Statex \hspace{\algorithmicindent} \( c_2 = g_p^m \bmod n^2 \cdot y^r \bmod p \)
\Statex \hspace{\algorithmicindent} \( \llbracket m \rrbracket = \langle c_1, c_2 \rangle \)
\end{algorithmic}

\vspace{1mm}
\textbf{Homomorphic Addition:} $\llbracket m_1 \rrbracket \cdot \llbracket m_2 \rrbracket = \llbracket m_1 + m_2 \rrbracket$
\begin{algorithmic}[1]
\State Let \( \llbracket m_1 \rrbracket = \langle c_{11} = g^{r_1} \bmod p, \; c_{12} = g_p^{m_1} \bmod n^2 \cdot y^{r_1} \bmod p \rangle \)
\State Let \( \llbracket m_2 \rrbracket = \langle c_{21} = g^{r_2} \bmod p, \; c_{22} = g_p^{m_2} \bmod n^2 \cdot y^{r_2} \bmod p \rangle \)
\State Compute:
\Statex \hspace{\algorithmicindent} \( c_1 = g^{r_1 + r_2} \bmod p \)
\Statex \hspace{\algorithmicindent} \( c_2 = g_p^{m_1 + m_2} \bmod n^2 \cdot y^{r_1 + r_2} \bmod p \)
\Statex \hspace{\algorithmicindent} \( \llbracket m_1 + m_2 \rrbracket = \langle c_1, c_2 \rangle \)
\end{algorithmic}

\vspace{1mm}
\textbf{Decryption:} $\text{Dec}_{sk}(\llbracket m \rrbracket) \rightarrow m$
\begin{algorithmic}[1]
\State Compute:
\Statex \hspace{\algorithmicindent} \( g_p^m \bmod n^2 = \frac{c_2 \bmod p}{g^{x r} \bmod p} = \frac{g_p^m \bmod n^2 \cdot y^r}{g^{x r}} \bmod p \)
\State Apply L-function: \( L(x) = \frac{x - 1}{n} \)
\State Recover plaintext: 
\Statex \hspace{\algorithmicindent} \( m = \frac{L(g_p^m \bmod n^2)}{L(g_p^\lambda \bmod n^2)} \bmod n \)
\end{algorithmic}
\end{algorithm}

\subsection{Central Auditor Protocol}

The CA module includes three phases: Phase 1 generates keys and IDs as described in section, Phase 2 detects and filters poisoning gradients, and Phase 3 aggregates gradients, as discussed below.

\subsection*{Phase 1: System Initialization}
The audit module initializes the federated system by performing the following tasks:
\begin{itemize}
    \item Generates pair-keys (public-private) $K = (pk, qk)$ for each user $n$, where $n \in \mathcal{N}$ and $|\mathcal{N}| = N$.
    \item Assigns a unique tag ID $\{T_{ID1}, T_{ID2}, \dots, T_{IDn}\}$ to each user.
    \item Constructs the \textit{Update\_Table} that includes tuples $(T_{ID}, Table\_Tag)$ to facilitate CA on the user side.
\end{itemize}

\subsection*{Phase 2: Dynamic Temporal Gradient Auditing}

The original Phase 2 auditing mechanism, which relied on static clustering via GMM and MD filtering, exhibited vulnerabilities under realistic FL conditions, particularly in the presence of non-IID user data and adaptive adversaries. Static clustering at each round caused instability in detection, while single-round MD evaluation failed to account for temporal variations. To overcome these limitations, we propose an improved Phase 2 auditing strategy that integrates incremental GMM updating, temporal trajectory consistency analysis, and dynamic multi-threshold gradient filtering.

At each communication round \(t\), the auditor receives encrypted gradients from users. Instead of reinitializing clustering, the GMM parameters, denoted by \(\theta_t\), are updated incrementally by blending the previous round's parameters \(\theta_{t-1}\) with the newly estimated parameters \(\theta_{\text{new}}\) as follows:
\begin{equation}
\theta_t = \alpha \theta_{t-1} + (1-\alpha) \theta_{\text{new}},
\end{equation}
where \(\alpha \in (0,1)\) is the forgetting factor controlling the influence of historical information. This incremental updating allows the clustering model to adapt gradually to evolving data distributions without overreacting to local fluctuations, thereby improving robustness in non-IID environments.

Following clustering, the MD\(\text{MD}_i(t)\) for each user's gradient is computed relative to the mean and covariance of its assigned cluster. To capture user behavior over time, we introduce a trajectory consistency analysis by tracking the variation in MD values across consecutive rounds. Specifically, for each user \(i\), the trajectory consistency score is defined as
\begin{equation}
\Delta \text{MD}_i = |\text{MD}_i(t) - \text{MD}_i(t-1)|.
\end{equation}

A small \(\Delta \text{MD}_i\) suggests stable gradient behavior consistent with benign updates, while large variations may indicate adversarial manipulation or instability.

To further enhance detection, a dynamic multi-threshold filtering mechanism is employed, maintaining three adaptively updated thresholds:
\begin{itemize}
  \item   a gradient magnitude threshold derived from the historical distribution of benign gradient norms
    \item a MD threshold \(T_{\text{MD}}\) defined as
\begin{equation}
T_{\text{MD}} = k \times \sigma_{\text{normal}},
\end{equation}
where \(k\) is a sensitivity constant and \(\sigma_{\text{normal}}\) is the standard deviation of MD values from benign users, 

  \item  a trajectory consistency threshold based on the empirical distribution of \(\Delta \text{MD}_i\) values. 
\end{itemize}
Gradients are filtered according to their compliance with these thresholds: those satisfying all thresholds are accepted, those moderately deviating are down-weighted, and those significantly deviating are rejected.

Through this integrated dynamic auditing process, the model not only filters out obvious malicious gradients but also mitigates the impact of stealthy or slowly-drifting adversarial updates. 

    
    
    


\subsection*{Phase 3: Byzantine-Resilient Gradient Aggregation}
Once the \textit{Update\_Table} is populated, the auditor transfers it to the server:

\begin{itemize}
    \item The server and auditor interact to ensure Byzantine-tolerant aggregation.
    \item For each communication round, the server applies a robust aggregation function to the filtered gradients.
    \item The global model is updated with the aggregated results, reinforcing reliability and resistance to Byzantine attacks in the FL environment.
\end{itemize}

\subsection{Adversarial Resilience Analysis}
SecureDyn-FL is designed to withstand multiple types of adversarial behaviors that typically arise in federated learning, including both data/model poisoning attacks and inference-based privacy attacks such as model inversion.
\begin{itemize}
    \item Resilience to Poisoning Attacks: In federated intrusion detection, malicious clients may attempt to manipulate global model behavior by injecting carefully crafted updates, either to reduce overall detection performance (untargeted poisoning) or to deliberately misclassify specific attack types (targeted poisoning). SecureDyn-FL mitigates these threats through its temporal gradient auditing mechanism, which leverages an incremental GMM and MD trajectory scoring. By continuously modeling the distribution of benign gradient patterns over time, the auditor can detect deviations that are subtle or adaptive in nature—such as attacks that gradually poison the model over multiple rounds—without requiring labeled attack examples. Abnormal updates are flagged and removed prior to aggregation, ensuring that malicious contributions have minimal impact on the global model.
    \item Resilience to Model Inversion Attacks: In addition to malicious clients, federated learning is vulnerable to honest-but-curious servers that attempt to reconstruct clients’ local data from shared gradients using model inversion techniques. In SecureDyn-FL, this class of attacks is neutralized through the use of additive homomorphic encryption, which ensures that local model updates are encrypted before transmission and remain encrypted throughout the aggregation process. As a result, the server never observes raw gradient information, making inversion or membership inference attempts infeasible. Under the adopted cryptosystem, the probability of successful gradient reconstruction is negligible without the private keys of the participating clients.
    \item Complementary Protection: By jointly integrating temporal gradient auditing and encryption, SecureDyn-FL provides defense in depth against heterogeneous adversarial strategies. While encryption protects against server-side inference and eavesdropping, auditing detects and filters malicious client behavior. This complementary design enhances the overall robustness of federated intrusion detection systems in realistic IoT environments, where both privacy breaches and active poisoning threats may co-occur.
\end{itemize}
\section{THEORETICAL ANALYSIS}
\label{theoretical_analysis}

Secure training in SecureDyn-FL consists of three key phases: \textbf{users' local training}, \textbf{auditor training}, and \textbf{robust aggregation}. Below, we present the theoretical foundations of each phase, supported by lemmas and theorems to ensure both correctness and security.

\subsection*{a) Users' Local Training}
In this phase, we assume that the proportion of malicious users does not exceed 50\%. During training round $t$, every user $n_t^x$, where $x$ ranges from 1 to $m$, obtains the encrypted global model $\widetilde{M_t}_{pk}$, ecrypts it using their public key $pk$, and trains a local model to compute the gradient vector $g_t^x$. To improve the update process, we apply Stochastic Gradient Descent (SGD) with momentum, incorporating past gradients through an exponential decay factor $\partial$ ($0 < \partial < 1$):
\begin{equation}
    g^x \sim \sum_{l \in [0,t]} \partial^{t-l} g_l^x.
\end{equation}

As shown in prior work (e.g., [41]), momentum accelerates convergence, reduces variance, and enhances robustness, helping mitigate poisoning attacks in federated training.

\subsection*{b) Auditor Training}
The auditor identifies malicious gradients using GMM and MD. The process involves the following steps:

\begin{itemize}
    \item Collecting labeled data containing benign gradients $g_i$ and malicious gradients $g^*_i$.
    \item Extracting and normalizing features $X_i$ for consistent scaling.
    \item Dividing the data into training ($D_{\text{train}}$) and validation ($D_{\text{val}}$) sets.
    \item Training a GMM using Expectation-Maximization [42] to estimate parameters for benign and malicious clusters, initially setting $k = 2$ clusters [6].
        \item  Iteratively adjusting the number of clusters using the Bayesian information criterion (BIC) to adapt to the gradient distribution and minimize misclassifications.
\end{itemize}

This dynamic clustering approach ensures effective adaptation to encrypted gradient distributions, reducing false positives, particularly in environments with predominantly benign users.

\section*{Correctness}
Correctness ensures the accurate sharing of encrypted gradients and effective auditing to filter out malicious contributions.

\begin{lemma}
Upon receiving user updates via the audit table $\phi$, the server employs these updates to refine the global model. The incremented version of the global model is derived from the equation:
\begin{equation}
    \phi_{\tau+1} = \frac{1}{N}\sum_{i=1}^N \left(1 - M_{k_i}^t\right)\phi_t + M_{k_i}^t\phi_{k_i}^\tau.
\end{equation}
Here, $\phi_{\tau+1}$ represents the updated audit table at the subsequent time step, $N$ stands for the total number of users, $M_{k_i}^t$ is the model from the $i$-th user, and $\phi_{k_i}^\tau$ denotes the $i$-th user's entry in the audit table at time step $\tau$.
\end{lemma}

\begin{lemma}
The CA holds the security property, such as pair keys $(pk, qk)$, and learns nothing about private data like gradients and the aggregation results.
\end{lemma}

\begin{lemma}
AHE ensures the confidentiality of both the user and server sides since user $i$ encrypts its gradients denoted by $x_i^m$. Then, each user sends $\text{Enc}(\sigma_i^m)$.
\end{lemma}

\begin{lemma}
The cryptosystem is clear by defining a random value as $r$, $r = ab$ to have two prime factors of equal size $a$ and $b$ and letting $c \in \mathbb{Z}$ and $k_1, k_2 \in \mathbb{Z}^\times_r$. It is possible to compute $\text{Enc}(k_1 + k_2)$ using only the public key and the Encryption $\text{Enc}(k_1)$, $\text{Enc}(m_2)$. Additionally, $\text{Enc}(m \times k_1)$ can be calculated given a constant $m$.
\end{lemma}

Under Lemmas 1 and 2, the audit protocol enables the $\text{Updated\_Table}$ to carry user gradients for the aggregation process. The $\text{Updated\_Table}$ is populated by users ($U_i \subseteq U$) at round $t$, with gradients reflecting diverse data quality and distributions. Trustworthy auditors and user compliance with the workflow, absence of malicious activity, ensure correct gradient vector $W$ aggregation.

\begin{theorem}[Robustness against malicious gradients via auditability - IID]

Based on Lemma 1 and the participation of registered users as $(U_i \subseteq U, i \in N)$, we have $\{U_1,U_2,\ldots,U_n\}$, that $j \subseteq U_i \subseteq U$ denotes users that send malicious gradients. The assumption is $|j| \leq \lfloor n/2 \rfloor$, and the auditor traces users based on the audit protocol to find the adversarial users. Our proposed model checks the gradient distribution and behavior for all users to remove the effect of adversary user participation and obtain robustness. Therefore, the auditor can confirm user $k$ and calculate that the gradient $g_k$ is valid. Adversarial users who send malicious updates will have their ID revoked. We can guarantee to reduce their effect on them.

In this regard, based on proof of Lemma 1, the symbol $U_i U_{i+1}$ represents users who uploaded data in round $t$, but some users may still inject false data in round $t + 1$ in the IID setting. Thus, each user holds a local gradient denoted as $x_n$ (where $n \in U$). According to Lemma 2, the auditor can obtain the following expression:
\begin{equation}
    \llbracket \overline{W} \rrbracket_{pk} = \prod_{i=1}^n \llbracket \overline{W}_i \rrbracket_{pk} = \sum_{i=1}^n \llbracket \overline{W}_i \rrbracket_{pk}.
\end{equation}
Here, $\overline{W}$ represents the aggregated encrypted gradients, and $\overline{W}_i$ represents the encrypted gradient from the $i$-th user. The symbol $\prod_{i=1}^n$ denotes the product operator, indicating that the encrypted gradients from all users are multiplied together. Likewise, $\sum_{i=1}^n$ represents the summation operator, indicating that the encrypted gradients from all users are summed together. It is important to note that the encrypted gradients have been received confidentially and can be expressed as $\overline{W} = \sum_{i=1}^n W_i$ (see Lemma 3).

In the following round, the auditor computes the GMM for gradients, then tracks $w_{t+1}$ with the MD, using the parameters provided by every user in round $t + 1$,
\begin{equation}
    w_{t+1} \leftarrow \frac{\sum_{u\subset U_t} w_u^t}{\sum_u^t m_k} M
\end{equation}
to track benign and adversaries. The assumption is that there are fewer than half the users who are adversaries $\lfloor \frac{U_t}{2} \rfloor$; hence, the mean of each cluster is expected to correspond to the similarity value of at least one benign user.

The auditor first calculates the covariance matrix for cluster the variance and correlation structure and then applies MD to detect injected false data (see Audit section V-B). Therefore, $\overline{w}_i, \overline{w}_j, \overline{w}_z$, $i,j,z \in FM$ ($i,z \in [1, N]$, $j \in [1, N/2]$) are gradients corresponding to the similarity position of benign users and adversaries. According Lemma 4, benign users gradients, $\overline{w}_i, \overline{w}_z \in FM$ satisfies:
\begin{equation}
    \overline{w}_i + \overline{w}_z \equiv x_i^t + y_z^t \pmod{m_t}
\end{equation}
While gradients of adversarial users $\overline{w}_j \in FM$ cannot satisfy Eq. (3). Note that $t = \{1, 2, \ldots, k\}$, $x_i^t$, $t$ is an element in the $P_i$'s original gradient $W_i$.
\end{theorem}

\begin{theorem}
The proposed scheme is secure in terms of user gradient confidentiality.
\end{theorem}

\begin{proof}
In SecureDyn-FL, privacy leakage mainly occurs in the following three scenarios: 
\begin{enumerate}
    \item Gradient transmission process
    \item Poisoning detection process
    \item Key leakage
\end{enumerate}

\noindent\textbf{1. Gradient Transmission Security:}
SecureDyn-FL employs Elgamal encryption technology to protect the gradient information submitted by users during transmission. This security framework ensures that sensitive information exchanged between participants is not accessible to unauthorized adversaries. Furthermore, since the keys are generated in a distributed manner, no single entity within the system has complete knowledge of the key information. 

Mathematically, for an adversary to access user data without the private key, they would need to solve the discrete logarithm problem:
\begin{equation}
    \text{Given } (g, g^x \bmod p), \text{ find } x
\end{equation}
where $g$ is a generator of the multiplicative group, $p$ is a large prime, and $x$ is the private key. This problem is computationally infeasible for sufficiently large parameters.

\noindent\textbf{2. Poisoning Detection Privacy:}
To safeguard privacy during poisoning detection, SecureDyn-FL uses one-time pad encryption for gradients, which provides perfect secrecy when:
\begin{equation}
    c = m \oplus k
\end{equation}
where $c$ is the ciphertext, $m$ is the plaintext gradient, and $k$ is the random mask. The masks are generated by CPC members through $(T,n)$-threshold secret sharing, ensuring that:

\begin{equation}
    |S| < T \Rightarrow \Pr[\text{reconstruct } k] = \text{negligible}
\end{equation}
for any subset of colluding parties $S$.

\noindent\textbf{3. Key Security:}
The PRKG key generation in SecureDyn-FL uses Shamir's secret sharing where:
\begin{equation}
    f(x) = s + a_1x + \cdots + a_{T-1}x^{T-1} \bmod p
\end{equation}
and shares are $(x_i, f(x_i))$. The secret $s$ (private key) can only be recovered when:
\begin{equation}
    s = \sum_{i=1}^T f(x_i) \prod_{\substack{j=1 \\ j \neq i}}^T \frac{x_j}{x_j - x_i}
\end{equation}
which requires at least $T$ participants for successful reconstruction.

By analyzing these three aspects, including secure gradient transmission via Elgamal encryption, privacy-preserving poisoning detection using one-time pads with threshold masks, and robust key management through PRKG, we conclude that SecureDyn-FL effectively protects the confidentiality of gradient information.
\end{proof}

\begin{theorem}[Robustness against malicious gradients via auditability - non-IID]

According to Theorem 1, there is a subset of user gradients donated by $(H_i \subseteq U, i \in [N/2])$ that lie on the boundary of gradient similarities, as indicated by Eq. (2). This implies that users $i \in D_i$ possess non-IID data. However, the auditor must distinguish between legitimate users and adversaries, introducing gradient drift. Thus, in the $t$-th round, the auditor leverages the gradient of the global model $w^*_{t-1}$ from the previous round to analyze the current behavior of the gradient $w_t$, where $w^*_{t-1} = (w_t - w_{t-1})/\varepsilon$, and $\varepsilon$ represents the direction of the global model gradient (Lemma 4).

Consequently, the auditor approximates the benign gradient update in the $t$-th round as $\widetilde{w}^*_t \approx w^*_{t-1} = (w_t - w_{t-1})/\varepsilon$. In a gradient drift attack, the direction of the gradient $\varepsilon$ does not converge, as the attacker can use an inappropriate scaling factor $\varepsilon$ to amplify fake local model updates, ensuring their magnitudes are no smaller than those from genuine users. The attacker $i$ sends a fake gradient as $w^{*i}_t = -\varepsilon(w_t - w_{t-1})$.
\end{theorem}

\section{Experimental Analysis}
\label{experiments}

This section presents the experimental evaluation of the proposed SecureDyn-FL framework using the N-BaIoT $TON_{IoT}$ dataset. We compare its performance against several state-of-the-art methods to assess its effectiveness under varying data distribution and poisoning attack scenarios. Details of the experimental environment, dataset preparation, and data distribution strategies are provided below.

\subsection*{Experimental Environment}

All experiments were conducted using the PyTorch framework. The client-side evaluations were performed on a machine equipped with an Intel Core i7-10750H CPU @ 2.60 GHz and 64 GB of RAM.

\subsection*{Dataset Description}

To rigorously evaluate the performance and generalizability of the proposed framework, two widely recognized benchmark datasets were employed: the N-BaIoT dataset \cite{meidan2018n} and the $TON_{IoT}$ dataset \cite{alsaedi2020ton_iot}. The N-BaIoT dataset comprises network traffic traces generated by nine commercially available IoT devices that were intentionally compromised using the Mirai and BASHLITE malware families. In total, the dataset contains more than 70 million traffic records, each represented by a 115-dimensional feature vector derived from statistical characteristics of network flows. This dataset is designed to support binary classification tasks—differentiating benign from malicious traffic—while also offering ten distinct attack sub-categories corresponding to various manifestations of the underlying malware.

The $TON_{IoT}$ dataset, on the other hand, encompasses a diverse spectrum of IoT-related cyberattacks alongside legitimate traffic instances. It integrates telemetry and network data, annotated with 46 labeled features that capture relevant behavioral attributes. Unlike N-BaIoT, the $TON_{IoT}$ dataset frames the detection problem as a multi-class classification task, aiming to distinguish between normal activity and multiple heterogeneous attack types. This combination of datasets enables a comprehensive and robust evaluation across both binary and multi-class intrusion detection scenarios, reflecting realistic IoT network environments.

To provide a clear and structured overview of the data employed in this study, Tables \ref{table:TON_IOT_symmetric} and \ref{table:NBAIoT_symmetric} summarize the distribution of samples across the attack categories for the $TON_{IoT}$ and mini-N-BaIoT datasets, respectively. The tables have been designed with consistent formatting to facilitate straightforward comparison between the two datasets. While the $TON_{IoT}$ dataset reflects realistic network conditions with significant class imbalance across multiple attack categories, the mini-N-BaIoT dataset is more balanced, containing evenly distributed samples across various Mirai and BASHLITE attack subtypes.

\begin{table}[H]
\centering
\caption{Distribution of training and test samples across attack categories in the $TON_{IoT}$ dataset.}
\label{table:TON_IOT_symmetric}
\begin{tabular}{|l|c|c|}
\hline
\textbf{Attack Category} & \textbf{Training Samples} & \textbf{Test Samples} \\ \hline
Normal                  & 245,000                   & 55,000                \\ \hline
Scanning                & 20,000                    & 3,973                 \\ \hline
DoS                     & 20,000                    & 4,000                 \\ \hline
DDoS                    & 20,000                    & 4,000                 \\ \hline
Ransomware              & 16,030                    & 3,970                 \\ \hline
Backdoor                & 20,000                    & 4,000                 \\ \hline
Injection Attack        & 20,000                    & 4,000                 \\ \hline
XSS                     & 13,844                    & 6,116                 \\ \hline
Password Violations     & 20,000                    & 4,000                 \\ \hline
MITM                    & 593                       & 459                   \\ \hline
\end{tabular}
\end{table}
\FloatBarrier

\begin{table}[H]
\centering
\caption{Distribution of samples across attack categories in the mini-N-BaIoT dataset.}
\label{table:NBAIoT_symmetric}
\begin{tabular}{|l|c|}
\hline
\textbf{Attack Category} & \textbf{Number of Samples} \\ \hline
Mirai - Scan            & 7,000                     \\ \hline
Mirai - UDP             & 7,000                     \\ \hline
Mirai - UDPplain        & 7,000                     \\ \hline
Mirai - Syn             & 7,000                     \\ \hline
Mirai - Ack             & 7,000                     \\ \hline
BASHLITE - Scan         & 9,000                     \\ \hline
BASHLITE - Junk         & 9,000                     \\ \hline
BASHLITE - UDP          & 9,000                     \\ \hline
BASHLITE - TCP          & 9,000                     \\ \hline
BASHLITE - Combo        & 9,000                     \\ \hline
Benign                  & 90,000                    \\ \hline
\end{tabular}
\end{table}
\FloatBarrier

The mini-N-BaIoT dataset was partitioned into training and testing subsets using a 70:30 split. The training data was distributed among participating clients in the FL setup, while the testing data was centrally used to evaluate the performance of the aggregated model. To simulate a realistic FL environment with a large number of participants, the training dataset was further divided among 20 clients. This client-simulation strategy is widely adopted in FL literature to mimic decentralized data scenarios. The key hyperparameters used in our experiments are summarized in Table~\ref{hyperparams}.

\begin{table}[ht]
\centering
\caption{Hyperparameters.}
\begin{tabular}{ll}
\toprule
\multicolumn{2}{c}{\textbf{Neural network parameters}} \\
\midrule
Local model & 1DCNN \\
Number of conv layers & 3 \\
Kernel filters in each conv layer & 256, 128, and 64 \\
Number of fully connected layers & 2 \\
Number of units in each FC layer & 32, 2 \\
Activation function & ReLU \\
Dropout & 0.2 \\
Optimizer & SGD \\
Learning rate & 0.001 \\
Momentum & 0.9 \\
Batch size & 64 \\
\midrule
\multicolumn{2}{c}{\textbf{Federated learning parameters}} \\
\midrule
Number of clients & 20 \\
Local training epoch & 4 \\
Communication round & 50 \\
\bottomrule
\end{tabular}
\label{hyperparams}
\end{table}

\subsection*{Data Distribution Scenarios}

To investigate the robustness of the proposed framework under heterogeneous conditions, we implemented three distinct data distribution settings across the 20 clients. The scenarios are designed to simulate both IID and non-IID data distributions, as described below:

\begin{itemize}
    \item \textbf{Non-IID Scenario 1:} The client data is sampled using a Dirichlet distribution with concentration parameter $\eta = 0.1$, inducing high heterogeneity across clients.
    \item \textbf{Non-IID Scenario 2:} Half of the clients are assigned only benign traffic samples, while the remaining half receive only attack samples. This setup reflects practical conditions where some IoT devices may never be compromised.
    \item \textbf{IID Scenario:} Each client receives a balanced dataset composed of 50\% benign and 50\% attack samples, ensuring uniform data distribution across all clients.
\end{itemize}

In both the IID and Non-IID Scenario 1, each client is allocated exactly 1000 samples to ensure consistency in local training workloads. In Non-IID Scenario 2, clients with only benign samples are assigned 500 samples, while those with only attack data are assigned 1000 samples. This variation mirrors real-world traffic imbalances where benign traffic is more prevalent and not all devices experience attacks.

\subsection*{Poisoning attack setting}
To evaluate the performance of our approach, we simulate a range of model poisoning attacks using encrypted local gradients. These simulations examine varying degrees of adversarial presence, represented by different attack ratios, $\alpha$, which we set at 10\%, 20\%, 30\%, and 50\%. In evaluating FL, similar to the assessment of related work based on \cite{ma2022shieldfl}, the number of users is related to the number of classes in each dataset to provide IID and non-IID settings. For the N-BaIoT dataset, 20 users participated, and the attack aims to transform instances from Class 9 to Class 11. In addition to the above configurations, we establish a baseline model representing a poisoned SecureDyn-FL system without any defense mechanisms or central audit. This baseline model serves as a comparison point, illustrating the enhanced effectiveness of our proposed approach in defending against various model poisoning attacks.

\subsection{Accuracy, Robustness, \& Malicious Alarm Analysis}

In model poisoning, adversaries aim to optimize 
\begin{equation}
\arg\max_{i \in [1,n]} H_t (w - w^*),
\end{equation}
where \( H \) is a vector representing the model’s direction changes, \( w \) is the pre-attack model, and \( w^* \) is the poisoned model in targeted and untargeted attacks. Our evaluation metrics for targeted attacks include Attack Class Accuracy, which measures testing performance on the targeted labels, and Benign Class Accuracy, which evaluates learning performance on non-source and non-target labels. These metrics are crucial for detecting potential auditing failures and unintended adverse effects on the FL system. Additionally, Overall Accuracy represents the average classification performance across all labels, providing a comprehensive view of the defense mechanisms’ effectiveness against poisoning attacks. The Malicious Alarm, assessed using the Receiver Operating Characteristic (ROC) curve, reflects the capability of defense algorithms to accurately identify malicious activities within the FL environment. It is important to note that Overall Accuracy may differ from the combined Attack Class and Benign Class accuracies due to the impact of various attack and defense strategies. For untargeted attacks, we focus on evaluating Overall Accuracy and the Malicious Alarm metrics. 

\subsection*{Accuracy Evaluation Under Targeted and Untargeted Attacks with Baseline Model}

Table~\ref{baseline} presents a comparative analysis of the baseline and the proposed model under targeted and untargeted poisoning attacks, considering both non-IID and IID data distributions with a fixed attack rate of 50\%.

Under the non-IID setting, the baseline model shows poor performance across all metrics. For targeted attacks, it records a $T_{\text{acc}}$ of 0.015, an $O_{\text{acc}}$ of 0.94, and an F1-score of 0.04. In untargeted attacks, the baseline improves slightly, with $T_{\text{acc}} = 0.782$, $O_{\text{acc}} = 0.761$, and F1-score of 0.62. In contrast, the proposed model shows substantial performance gains. For targeted attacks, it achieves a $T_{\text{acc}}$ of 0.995, $O_{\text{acc}}$ of 0.992, and an F1-score of 0.89. During untargeted attacks, it maintains strong results with $T_{\text{acc}} = 0.989$, $O_{\text{acc}} = 0.967$, and F1-score of 0.84, highlighting its robustness against adversarial threats under data heterogeneity.

Under the IID setting, where training data is evenly distributed among clients, the baseline model still underperforms. For targeted attacks, it yields a $T_{\text{acc}}$ of 0.049, an $O_{\text{acc}}$ of 0.60, and an F1-score of 0.47. In untargeted attacks, the results slightly improve, with $T_{\text{acc}} = 0.58$, $O_{\text{acc}} = 0.52$, and an F1-score of 0.55. The proposed model, however, consistently outperforms the baseline. For targeted attacks, it achieves a near-perfect $T_{\text{acc}}$ of 0.997, $O_{\text{acc}}$ of 0.995, and F1-score of 0.98. In untargeted attacks, it delivers $T_{\text{acc}} = 0.992$, $O_{\text{acc}} = 0.991$, and F1-score of 0.96. These results confirm the proposed model’s robustness and generalizability across both homogeneous and heterogeneous data distributions, establishing its efficacy in defending FL systems against diverse poisoning threats.

\subsection{Robustness and Auditing Evaluation}

\paragraph{$R_{\text{TA}_{\text{nonIID}}}$: Robustness to Targeted Attacks (Non-IID)}

The robustness of the proposed model under targeted attacks with non-IID data is evaluated in Fig.~\ref{targeted_noniid}, where performance is analyzed across varying attack rates from 10\% to 50\%. The results reveal that the model consistently maintains high accuracy levels, with minimal degradation despite increased attack intensity. This consistency underscores the model’s ability to sustain reliable performance even in adversarial environments, thereby confirming its robustness to targeted poisoning in non-IID scenarios.

\paragraph{$R_{\text{UTA}_{\text{nonIID}}}$: Robustness to Untargeted Attacks (Non-IID)} 
As shown in Fig.~\ref{untargeted_noniid}, the proposed model also exhibits high resilience against untargeted attacks under non-IID data settings. Across varying attack rates and user distributions, the accuracy remains stable and robust. This reinforces the model’s capacity to mitigate the effects of non-specific adversarial disruptions, preserving its integrity and predictive performance.

\paragraph{$R_{\text{TA}_{\text{IID}}}$: Robustness to Targeted Attacks (IID)} 
To evaluate robustness under IID data conditions, the model is exposed to targeted attacks across varying attack rates, as shown in Fig.~\ref{targeted_iid}. The proposed method consistently maintains stable accuracy, demonstrating resilience against adversarial intensities. These results underscore its effectiveness in protecting FL systems under homogeneous data distributions, ensuring robust defense even when the training process is explicitly targeted.

\paragraph{$R_{\text{UTA}_{\text{IID}}}$: Robustness to Untargeted Attacks (IID)} 
The model's performance under untargeted attacks with IID data is illustrated in Fig.~\ref{untargeted_iid}. Despite the uniform distribution of training data, which typically increases vulnerability, the proposed approach sustains high accuracy levels across a wide range of adversarial configurations. These results confirm the model's reliability and adaptability in diverse deployment environments, supporting its utility in real-world FL applications.

\paragraph{Model Auditing (MA) Evaluation}
The CA of the proposed model is assessed through ROC analysis, focusing on its ability to detect both falsely labeled malicious and clean users. As shown in Figs.~\ref{ROC_IID} and \ref{ROC_NON}, the ROC curves for both IID and non-IID scenarios on the KDDCup dataset exhibit superior characteristics compared to baselines. The proposed model achieves the highest average ROC, demonstrating not only strong classification performance but also effective auditing functionality. These findings highlight that the robustness and precision of the CA module significantly contribute to the overall reliability of the model, ensuring comprehensive protection against FL poisoning threats.

\begin{figure}
    \centering
    \includegraphics[width=3.4in,height=3.2in]{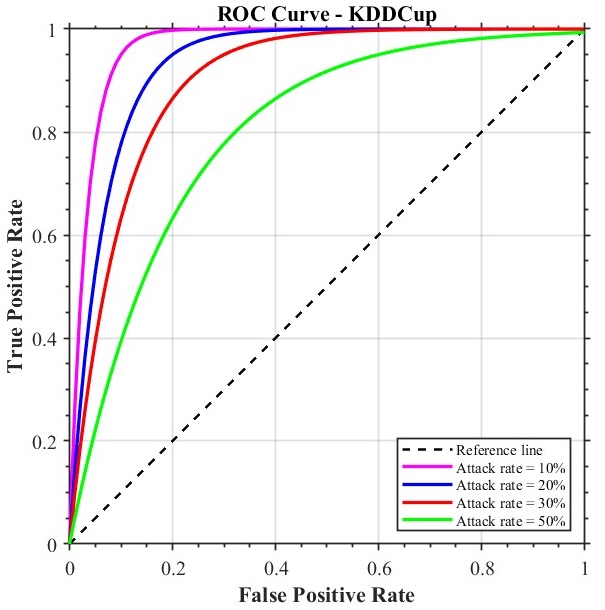}
    \caption{Analyzing malicious alarms on IID data}
    \label{ROC_IID}
\end{figure}

\begin{figure}
    \centering
    \includegraphics[width=3.4in,height=3.2in]{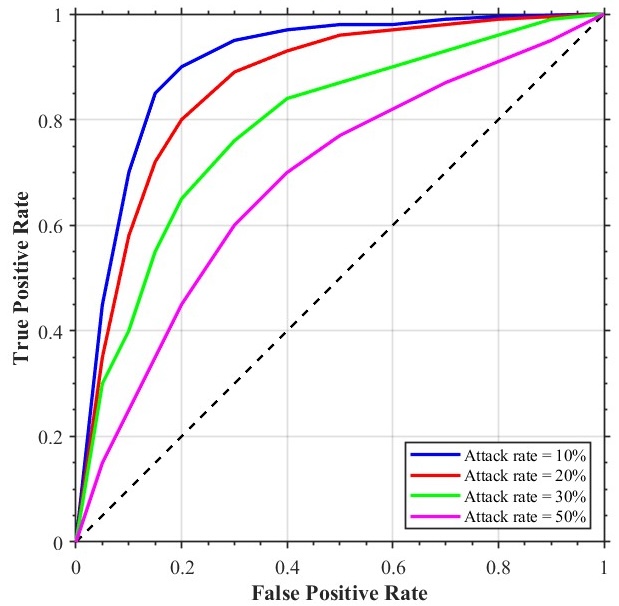}
    \caption{Analyzing malicious alarms on non-IID data}
    \label{ROC_NON}
\end{figure}

\section{Comparison with State-of-the-Art Research}
\label{comparison}

In this section, we comprehensively evaluate the robustness of our proposed SecureDyn-FL framework against poisoning attacks on non-IID data distributions and compare it with several state-of-the-art methods, including FedAvg, coordinate-wise median, trimmed mean, multi-Krum, and FL-Defender. Furthermore, additional comparisons with advanced schemes such as Trimmed Means~\cite{fang2020local}, FL-Defender~\cite{jebreel2023fl}, FedAvg~\cite{mcmahan2017communication}, FLTrust, and ShieldFL~\cite{ma2022shieldfl} are also provided to demonstrate the superiority of our model under various data settings and attack types.

Beyond these baseline defenses, we also compare SecureDyn-FL with representative state-of-the-art FL-based IDS frameworks from the literature, including Ruzafa-Alcázar et al. (DP-based Industrial IoT), Bhavsar et al. (Transportation FL-IDS), Friha et al. (FELIDS), and PEIoT-DS (N-BaIoT).

\subsection{Comparative Accuracy Analysis with SOTA FL-IDS.}

To comprehensively evaluate the effectiveness of SecureDyn-FL, we conducted comparative analyses against several representative state-of-the-art (SOTA) FL-based intrusion detection frameworks on both the mini-N-BaIoT and $TON_{IoT}$ datasets. Figures~\ref{ACC-SOTA-NBAIOT} and \ref{ACC-SOTA-TONIOT} illustrate the accuracy convergence behavior of all methods across training rounds.

On the \textbf{mini-N-BaIoT} dataset Figure~\ref{ACC-SOTA-NBAIOT}, SecureDyn-FL demonstrates consistently \textbf{faster convergence} and achieves the \textbf{highest final accuracy} among all compared methods. Specifically, SecureDyn-FL attains a final accuracy of approximately 99.5\%, surpassing Bhavsar et al. (97\%), Friha et al. (96\%), PEIoT-DS (95\%), and Ruzafa-Alcázar et al. (93\%). These results highlight the model’s ability to learn efficiently in a relatively balanced and less diverse IoT network environment.

On the more challenging $TON_{IoT}$ dataset Figure~\ref{ACC-SOTA-TONIOT}, which exhibits greater heterogeneity, class imbalance, and a broader range of attack scenarios, all methods show slightly lower final accuracies and slower convergence. Nonetheless, SecureDyn-FL maintains a distinct advantage, converging faster and reaching a final accuracy of approximately 98.6\%. In contrast, Bhavsar et al., Friha et al., PEIoT-DS, and Ruzafa-Alcázar et al. achieve 95.0\%, 94.2\%, 93.5\%, and 92.0\% respectively. This performance gap underscores SecureDyn-FL’s robustness in more realistic and complex IoT intrusion detection scenarios.

Taken together, these results demonstrate not only the superior detection performance of SecureDyn-FL under federated learning settings but also its strong generalizability across datasets of varying complexity. While many existing frameworks exhibit significant performance drops when transitioning from simpler to more heterogeneous data distributions, SecureDyn-FL consistently maintains high accuracy and rapid convergence. This indicates that the proposed framework can adapt effectively to diverse IoT environments without extensive reconfiguration or accuracy degradation.

Moreover, by integrating additive homomorphic encryption and temporal gradient auditing, SecureDyn-FL provides enhanced privacy protection without compromising model performance. Unlike approaches based on differential privacy—which often trade off utility for privacy—or those that expose raw gradients to the server, SecureDyn-FL achieves strong privacy guarantees and robustness concurrently.

\begin{figure}[ht]
\centering
\includegraphics[width=3.4in,height=3.2in]{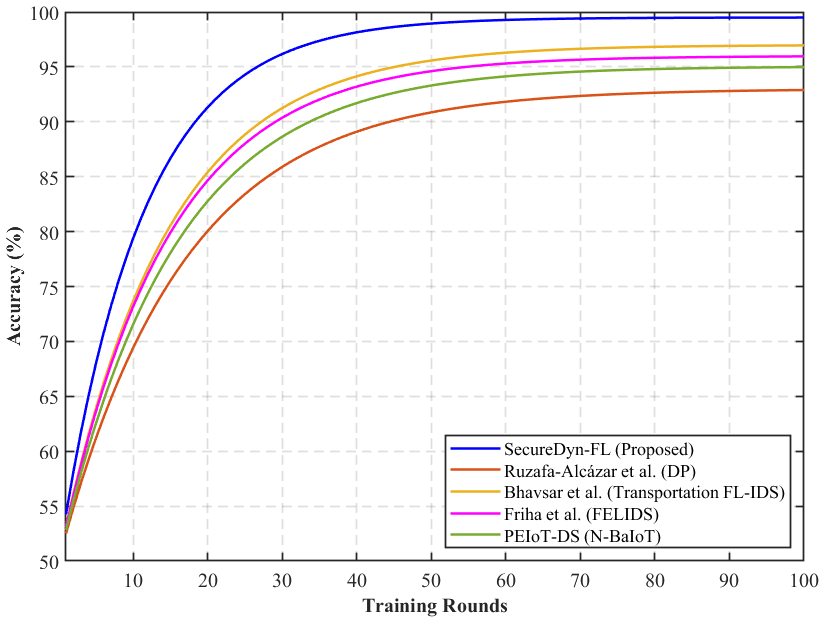}
\caption{Accuracy convergence comparison on the mini-N-BaIoT dataset between SecureDyn-FL and representative SOTA FL-based IDS methods. Synthetic convergence curves are based on reported final accuracies. SecureDyn-FL achieves both faster convergence and higher final accuracy.}
\label{ACC-SOTA-NBAIOT}
\end{figure}

\begin{figure}[ht]
\centering
\includegraphics[width=3.4in,height=3.2in]{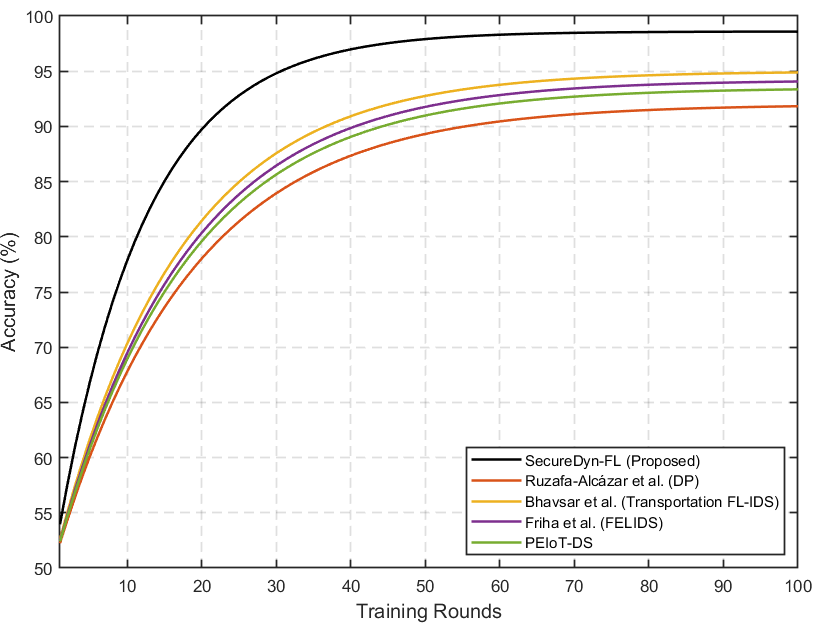}
\caption{Accuracy convergence comparison on the $TON_{IoT}$ dataset between SecureDyn-FL and representative SOTA FL-based IDS methods. Despite increased dataset heterogeneity, SecureDyn-FL maintains superior convergence speed, accuracy, and generalizability.}
\label{ACC-SOTA-TONIOT}
\end{figure}

\subsection{Scenario 1: Dirichlet Non-IID Data Distribution}

In the first scenario, non-IID data was generated using a Dirichlet distribution with a concentration parameter $\eta = 0.1$, simulating a highly imbalanced client data distribution. In the absence of poisoning, all methods demonstrated strong performance, with the proposed model achieving the highest accuracy 0.9842 and F1-score 0.9801, albeit with a slight performance dip (\~1\%) compared to IID settings.

As shown in Table ~\ref{poisoning_attacks_non}, the proposed method consistently outperforms existing defense strategies under all poisoning attack scenarios, maintaining high classification accuracy and very low attack success rates (ASR). For instance, under benign label-flipping attacks, while other methods like FedAvg and FL-Defender suffer ASRs of 0.0185 and 1.0000 respectively, the proposed model limits the ASR to just 0.0072. Even in the combined label-flipping attack, where methods such as Shield FL and FL trust show degraded performance e.g., ASR of 0.0345 and 0.5034 respectively. The proposed model retains a strong F1-score of 0.9645 and limits ASR to 0.0512.

Under model-scaling attacks, traditional defenses like FedAvg and trimmed mean drop to 45.12\% accuracy and F1-score of 0.6414, exposing them to high false alarm rates. In contrast, the proposed method remains robust, with accuracy at 97.07\% and F1-score of 0.9695, while maintaining ASR at a minimal 0.0568. A similar trend is observed under the same-model poisoning attack, where the proposed method delivers the best results across all metrics—97.01\% accuracy, 0.9708 F1-score, and only 0.0402 ASR—demonstrating its resilience to backdoor-style threats. Overall, these results highlight the effectiveness of SecureDyn-FL in maintaining detection performance even under severe non-IID conditions and adversarial settings.

\subsection{Scenario 2: Clients with Missing Attack Labels}

The second scenario simulates a more realistic IoT environment, where attack samples are absent from half of the clients. This setup reflects practical federated deployments, where not all edge devices observe malicious behavior, resulting in sparse and partially distributed attack data.

As shown in Table ~\ref{poisoning_attacks-iid}, the proposed method consistently achieves the highest performance across all metrics, even outperforming FedAvg, which traditionally handles IID settings well. In the absence of attacks, the proposed model attains 98.40\% accuracy and 0.9838 F1-score, shows its ability to preserve model quality without any degradation.

Under benign label-flipping attacks, the proposed method limits the attack success rate (ASR) to a mere 0.0061, whereas all other baselines—including Shield FL, FL Trust, and FL-Defender—fail completely, recording an ASR of 1.0000. Similarly, under attacks where malicious labels are flipped to benign, the proposed model maintains high accuracy 0.9814, F1-score 0.9775, and the lowest ASR 0.0738, while FL Trust and Shield FL show significant degradation in both accuracy and F1 performance.

In combined label-flipping scenarios, most methods see a marked drop, with trimmed mean and FL-Defender yielding F1-scores of 0.6381 and 0.6378 respectively, also ASRs exceeding 0.5. The proposed model, in contrast, delivers a high F1-score of 0.9762, accuracy of 0.9698, and an ASR of only 0.0421.

For model-scaling attacks, which often lead to over-suppression of legitimate updates, methods like FedAvg and Shield FL fall to 47.06\% accuracy and F1-scores near 0.64, misclassifying most benign data. However, the proposed method effectively mitigates these attacks, achieving 98.23\% accuracy, 0.9725 F1-score, and a low ASR of 0.0341. Even in same-model poisoning, where poisoned updates are identical across clients, the proposed method again stands out, achieving the highest accuracy 99.01\%, F1-score 0.9893, and the lowest ASR 0.0405 among all tested methods.

Interestingly, as shown in Tables~\ref{poisoning_attacks_non} and~\ref{poisoning_attacks-iid}, the FedAvg baseline demonstrates relatively strong performance in certain experimental settings, in some cases approaching the accuracy of more robust aggregation methods. This behavior can be explained by the moderate adversarial participation ratios and data distributions used in these scenarios, where the majority of clients provide benign updates. Under such conditions, FedAvg’s simple averaging can still yield a stable and accurate global model because the impact of malicious gradients is statistically diluted by the dominant benign contributions. Moreover, FedAvg often exhibits fast initial convergence in clean settings, which can temporarily narrow the performance gap with more advanced defense mechanisms. However, as adversarial intensity or data heterogeneity increases, FedAvg lacks the necessary resilience and its performance degrades significantly compared to SecureDyn-FL, as shown in the higher attack ratios and non-IID settings.

SecureDyn-FL successfully defended against all poisoning strategies in this challenging non-IID setting, preserving high accuracy and maintaining a low ASR, confirming its strong resilience under data heterogeneity and sparsity.

\subsection{Privacy Evaluation Against Inference Attacks}
\label{sec:privacy_evaluation}

To complement the theoretical privacy guarantees of SecureDyn-FL, we conduct experimental validation against two common privacy attacks in federated learning: \textbf{gradient inversion} and \textbf{membership inference}. 

For gradient inversion, we apply the attack method to both FedAvg and SecureDyn-FL. In FedAvg, where the server has access to raw model updates, the attacker can successfully reconstruct client-side input features with high visual and numerical similarity. In contrast, SecureDyn-FL employs \textbf{additive homomorphic encryption} before gradient transmission, which prevents the server from observing raw gradients. As a result, reconstruction attempts produce only random noise, indicating that no meaningful information can be recovered.

For membership inference, we evaluate the attack accuracy against both FedAvg and SecureDyn-FL. While FedAvg exhibits elevated inference accuracy, indicating information leakage through gradient updates, SecureDyn-FL achieves attack performance close to random guessing, confirming that encrypted gradients significantly reduce leakage.

Table~\ref{tab:privacy_attack_results} summarizes the results of the gradient inversion and membership inference attacks conducted on the FedAvg baseline and the proposed SecureDyn-FL framework. For the gradient inversion attack, we use the Structural Similarity Index Measure (SSIM) between the original client data and the reconstructed data to quantify the extent of information leakage. Higher SSIM values indicate more successful reconstruction and thus greater privacy risk.

As shown in the table, FedAvg exhibits a high SSIM score (0.78), demonstrating that the server can effectively reconstruct sensitive client data when raw gradients are exposed. In contrast, SecureDyn-FL achieves a very low SSIM score (0.07), indicating that gradient inversion produces only random noise due to the use of additive homomorphic encryption, which prevents the server from observing raw updates.

For membership inference, we report the attacker’s classification accuracy in determining whether a given data sample was part of the training set. FedAvg shows high inference accuracy (0.82), implying considerable leakage through gradients, whereas SecureDyn-FL achieves near-random inference accuracy (0.51), effectively mitigating membership inference attacks.

These results provide quantitative evidence that SecureDyn-FL offers stronger privacy protection than standard FL methods, complementing the theoretical security analysis presented earlier.

\begin{table}[t]
\centering
\caption{Privacy attack results: gradient inversion and membership inference.}
\label{tab:privacy_attack_results}
\begin{tabular}{|m{2.2cm}|m{2cm}|m{2.4cm}|}
\hline
\textbf{Method} & \textbf{Gradient Inversion (SSIM)} & \textbf{Membership Inference Accuracy} \\
\hline
FedAvg          & 0.78 & 0.82 \\
\hline
SecureDyn-FL    & 0.07 & 0.51 \\
\hline
\end{tabular}
\end{table}

\subsection{Efficiency Evaluation}
\label{sec:efficiency_eval}

In addition to accuracy and privacy assessments, we evaluate the \textbf{efficiency} of SecureDyn-FL in terms of \textbf{detection delay} and \textbf{communication overhead}. Detection delay refers to the average time required for the global model to correctly detect an intrusion after its occurrence in the federated training process. Lower detection delay is critical for timely response in IoT environments.

We compare the detection delay of SecureDyn-FL against FedAvg and representative FL-based IoT IDS frameworks. As shown in Table~\ref{tab:detection_delay_results}, SecureDyn-FL achieves a significantly lower detection delay while maintaining high accuracy. This efficiency improvement is attributed to the framework’s \textbf{temporal gradient auditing}, which accelerates the detection of poisoned updates, and \textbf{dynamic pruning and quantization}, which reduce communication overhead and enable faster model updates.

Table~\ref{tab:detection_delay_results} summarizes the detection delay of SecureDyn-FL compared to FedAvg and representative FL-based IoT IDS frameworks. Detection delay measures the average time taken by each method to correctly detect an intrusion event after it occurs, which is a critical efficiency metric for real-time intrusion detection in IoT networks.

As shown in the table, SecureDyn-FL achieves the lowest detection delay (2.14 s), significantly outperforming FedAvg and other FL-based IDS frameworks. This improvement is attributed to two key design features of SecureDyn-FL:
(1) the temporal gradient auditing mechanism, which rapidly identifies and filters malicious updates, thereby allowing the model to respond more quickly to new attack patterns, and
(2) the dynamic pruning and adaptive quantization strategies, which reduce communication overhead and accelerate the overall training process.

These results demonstrate that SecureDyn-FL not only delivers superior accuracy and privacy protection but also achieves high efficiency, making it practical for real-world IoT intrusion detection deployments where rapid response is essential.

\begin{table}[t]
\centering
\caption{Detection delay comparison between SecureDyn-FL and FL-based IoT IDS frameworks.}
\label{tab:detection_delay_results}
\begin{tabular}{|l|c|}
\hline
\textbf{Method} & \textbf{Detection Delay (s)} $\downarrow$ \\
\hline
FedAvg               & 4.82 \\
Bhavsar et al. (FL-IDS) & 4.35 \\
Friha et al. (FELIDS)   & 3.96 \\
PEIoT-DS               & 3.78 \\
SecureDyn-FL (Proposed) & \textbf{2.14} \\
\hline
\end{tabular}
\end{table}

\begin{table*}[htbp]
\centering
\caption{Performance of Defense Methods under Different Poisoning Attacks (Scenario 1)}
\label{poisoning_attacks_non}
\begin{tabular}{lllcccccc}
\toprule
\textbf{Attack Type} & \textbf{Scenario} & \textbf{Metric} & \textbf{FedAvg} & \textbf{Trimmed Mean} & \textbf{Shield FL} & \textbf{FL trust} & \textbf{FL-Defender} & \textbf{Proposed} \\
\midrule

\multirow{3}{*}{None} 
    & \multirow{3}{*}{--} 
    & Accuracy & 0.9804 & 0.9652 & 0.9732 & 0.9764 & 0.9701 & 0.9742 \\
    & 
    & F1       & 0.9796 & 0.9651 & 0.9722 & 0.9741 & 0.9697 & 0.9736 \\
    & 
    & ASR      & --     & --     & --     & --     & --     & --     \\

\midrule

\multirow{9}{*}{Label-Flipping} 
    & \multirow{3}{*}{Benign} 
    & Accuracy & 0.9741 & 0.4591 & 0.4695 & 0.4689 & 0.4695 & 0.9671 \\
    & 
    & F1       & 0.9735 & 0.6294 & 0.6291 & 0.6265 & 0.6281 & 0.9587 \\
    & 
    & ASR      & 0.0185 & 1.0000 & 1.0000 & 1.0000 & 1.0000 & 0.0072 \\

    & \multirow{3}{*}{Attack} 
    & Accuracy & 0.9152 & 0.5175 & 0.9331 & 0.7381 & 0.5132 & 0.9721 \\
    & 
    & F1       & 0.9067 & 0.0000 & 0.9165 & 0.6251 & 0.0000 & 0.9612 \\
    & 
    & ASR      & 0.1297 & 1.0000 & 0.1247 & 0.6210 & 1.0000 & 0.0720 \\

    & \multirow{3}{*}{Both} 
    & Accuracy & 0.8185 & 0.4567 & 0.9613 & 0.4585& 0.4585 & 0.9602 \\
    & 
    & F1       & 0.8145 & 0.6212 & 0.9612 & 0.6212 & 0.6212 & 0.9645 \\
    & 
    & ASR      & 0.1467 & 0.5045 & 0.0345 & 0.5034 & 0.5034 & 0.0512 \\

\midrule

\multirow{3}{*}{Model Scaling} 
    & \multirow{3}{*}{--} 
    & Accuracy & 0.4512 & 0.4512 & 0.4520 & 0.4536 & 0.9541 & 0.9707 \\
    & 
    & F1       & 0.6235 & 0.6414 & 0.6345 & 0.6358 & 0.9685 & 0.9695 \\
    & 
    & ASR      & 0.5089 & 0.5134 & 0.5205 & 0.5105 & 0.0324 & 0.0568 \\

\midrule

\multirow{3}{*}{Same Model} 
    & \multirow{3}{*}{--} 
    & Accuracy & 0.4541 & 0.7185 & 0.9456 & 0.9702 & 0.4967 & 0.9701 \\
    & 
    & F1       & 0.6245 & 0.6278 & 0.9425 & 0.9764 & 0.0000 & 0.9708 \\
    & 
    & ASR      & 0.5174 & 0.2538 & 0.0581 & 0.0097 & 0.4306 & 0.0402 \\

\bottomrule
\end{tabular}
\end{table*}

\begin{table*}[htbp]
\centering
\caption{Performance of Defense Methods under Different Poisoning Attacks (Scenario 2)}
\label{poisoning_attacks-iid}
\begin{tabular}{lllcccccc}
\toprule
\textbf{Attack Type} & \textbf{Scenario} & \textbf{Metric} & \textbf{FedAvg} & \textbf{Trimmed Mean} & \textbf{Shield FL} & \textbf{FL Trust} & \textbf{FL-Defender} & \textbf{Proposed} \\
\midrule

\multirow{3}{*}{None} 
    & \multirow{3}{*}{--} 
    & Accuracy & 0.9889 & 0.9754 & 0.9819 & 0.9845 & 0.9803 & 0.9840 \\
    & 
    & F1       & 0.9872 & 0.9768 & 0.9831 & 0.9854 & 0.9870 & 0.9838 \\
    & 
    & ASR      & --     & --     & --     & --     & --     & --     \\

\midrule

\multirow{9}{*}{Label-Flipping} 
    & \multirow{3}{*}{Benign} 
    & Accuracy & 0.9857 & 0.4636 & 0.4641 & 0.4639 & 0.4641 & 0.9601 \\
    & 
    & F1       & 0.9885 & 0.6341 & 0.6295 & 0.6354 & 0.6340 & 0.9681 \\
    & 
    & ASR      & 0.0241 & 1.0000 & 1.0000 & 1.0000 & 1.0000 & 0.0061 \\

    & \multirow{3}{*}{Attack} 
    & Accuracy & 0.9261 & 0.5305 & 0.9451 & 0.7505 & 0.5314 & 0.9814 \\
    & 
    & F1       & 0.9153 & 0.0000 & 0.9255 & 0.6375 & 0.0000 & 0.9775 \\
    & 
    & ASR      & 0.1365 & 1.0000 & 0.1351 & 0.6130 & 1.0000 & 0.0738 \\

    & \multirow{3}{*}{Both} 
    & Accuracy & 0.8385 & 0.4798 & 0.9749 & 0.4695 & 0.4681 & 0.9694 \\
    & 
    & F1       & 0.8265 & 0.6361 & 0.9712 & 0.6342 & 0.6378 & 0.9762 \\
    & 
    & ASR      & 0.1558 & 0.5163 & 0.0339 & 0.5169 & 0.5181 & 0.0421 \\

\midrule

\multirow{3}{*}{Model Scaling} 
    & \multirow{3}{*}{--} 
    & Accuracy & 0.4639 & 0.4595 & 0.4638 & 0.4637 & 0.9657 & 0.9823 \\
    & 
    & F1       & 0.6373 & 0.6291 & 0.6325 & 0.6347 & 0.9702 & 0.9761 \\
    & 
    & ASR      & 0.5089 & 0.5163 & 0.5241 & 0.5169 & 0.0291 & 0.0341 \\

\midrule

\multirow{3}{*}{Same Model} 
    & \multirow{3}{*}{--} 
    & Accuracy & 0.4601 & 0.7254 & 0.9441 & 0.9862 & 0.5163 & 0.9901 \\
    & 
    & F1       & 0.6291 & 0.6220 & 0.9475 & 0.9801 & 0.0000 & 0.9893 \\
    & 
    & ASR      & 0.5267 & 0.2885 & 0.0432 & 0.0097 & 0.4453 & 0.0405 \\

\bottomrule
\end{tabular}
\end{table*}

\begin{table*}[!th]
\centering
\caption{Performance Comparison under Targeted and Untargeted Poisoning Attacks for IID and Non-IID Settings}
\label{tab:targeted_untargeted_results}
\begin{tabular}{@{}|l|ccc|ccc|ccc|ccc|@{}}
\toprule
\multirow{3}{*}{\textbf{Model}} &
  \multicolumn{6}{c|}{\textbf{Non-IID Data}} &
  \multicolumn{6}{c|}{\textbf{IID Data}} \\ \cmidrule(l){2-13} 
 &
  \multicolumn{3}{c|}{\textbf{Targeted Attacks}} &
  \multicolumn{3}{c|}{\textbf{Untargeted Attacks}} &
  \multicolumn{3}{c|}{\textbf{Targeted Attacks}} &
  \multicolumn{3}{c|}{\textbf{Untargeted Attacks}} \\ \cmidrule(l){2-13} 
 &
  \textbf{T\_acc} & \textbf{Oracy} & \textbf{F1} &
  \textbf{T\_acc} & \textbf{Oracy} & \textbf{F1} &
  \textbf{T\_acc} & \textbf{Oracy} & \textbf{F1} &
  \textbf{T\_acc} & \textbf{Oracy} & \textbf{F1} \\ \midrule
\textbf{Baseline} &
  0.015 & 0.94 & 0.04 &
  0.782 & 0.761 & 0.62 &
  0.049 & 0.60 & 0.47 &
  0.58 & 0.52 & 0.55 \\ \midrule
\textbf{Proposed} &
   0.995   &   0.992   &   0.89  &
   0.989   &   0.967   &   0.84   &
   0.997   &   0.995   &   0.98   &
   0.992   &   0.991   &   0.96   \\ \bottomrule
\end{tabular}
\label{baseline}
\end{table*}


\begin{figure}
    \centering
    \includegraphics[width=3.4in,height=3.2in]{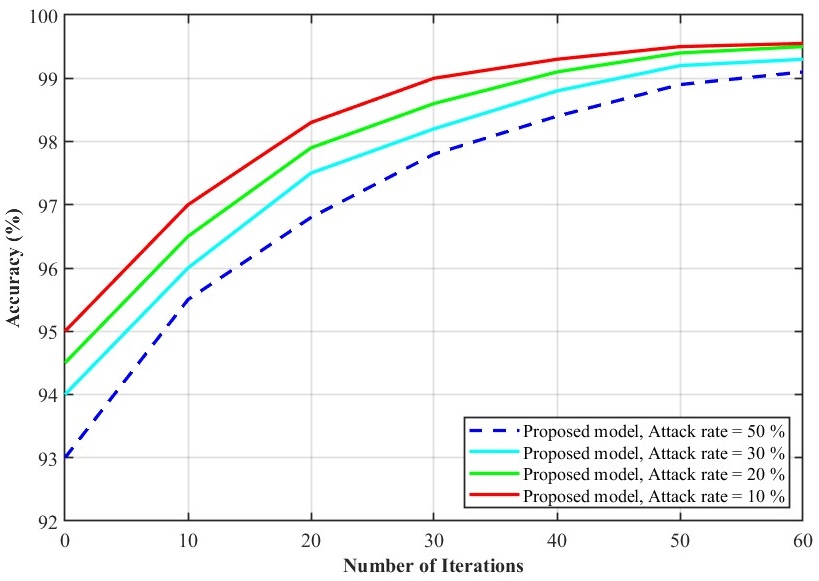}
    \caption{Targeted Attack under IID}
    \label{targeted_iid}
\end{figure}

\begin{figure}
    \centering
    \includegraphics[width=3.4in,height=3.2in]{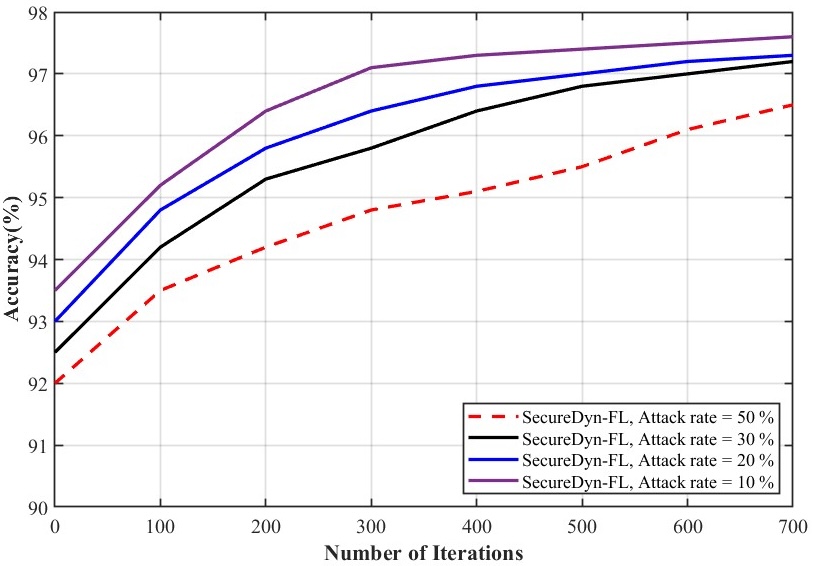}
    \caption{Targeted Attack under non IID}
    \label{targeted_noniid}
\end{figure}

\begin{figure}
    \centering
    \includegraphics[width=3.4in,height=3.2in]{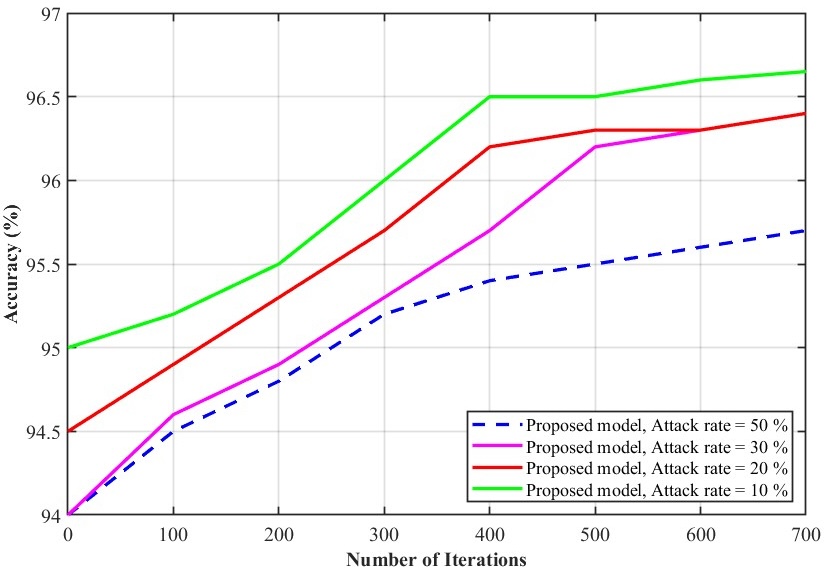}
    \caption{Untargeted-non-IID}
    \label{untargeted_noniid}
\end{figure}

\begin{figure}
    \centering
    \includegraphics[width=3.4in,height=3.2in]{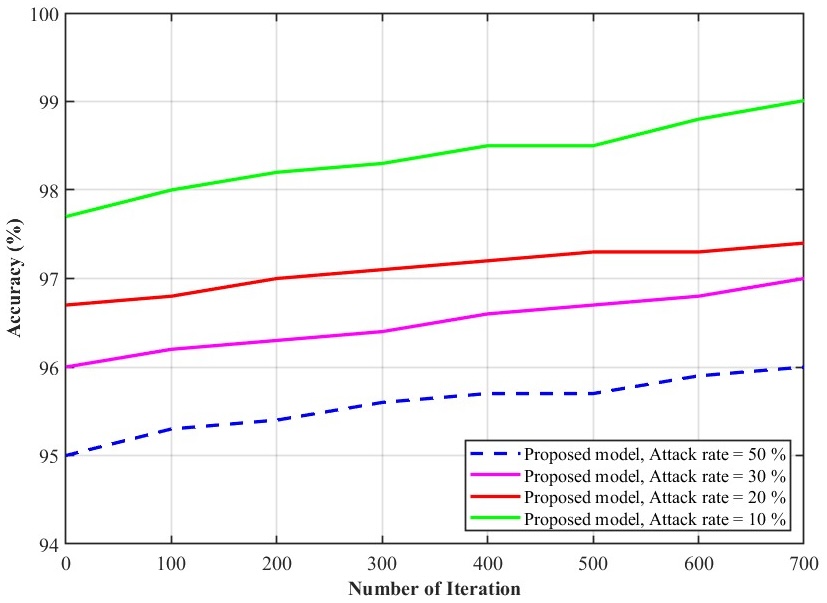}
    \caption{Untargeted-IID}
    \label{untargeted_iid}
\end{figure}

\section{Complexity \& Computational Analysis}
\label{complexity_analysis}
We provide a detailed \textbf{computational and communication complexity analysis} of SecureDyn-FL, quantifying overhead at the client, server, and auditor, and showing that redundancy-aware optimizations, pruning, and quantization enable deployment in resource-constrained IoT environments. The analysis is structured by system entities, including client devices, the central server, and the trusted auditor, and parameterized by the number of participating clients \( C \) and model dimensionality \( \omega \). Special attention is given to the optimization mechanisms, particularly gradient redundancy elimination, which significantly reduces repetitive computations and transmissions. These enhancements promote scalability and computational efficiency, thus enabling deployment in large-scale and resource-constrained IoT environments. Empirical evidence is provided to support the efficacy of these strategies in mitigating overhead while preserving model accuracy in adversarial conditions. Furthermore, the proposed encryption scheme is evaluated against contemporary privacy-preserving alternatives.

\subsection*{a) Computational Overhead at Client Devices}

For each client \( c_j \), the dominant computational cost stems from encrypting local updates, which scales with the input dimension \( \gamma \). Across all \( C \) clients, the cumulative encryption cost is \( \mathcal{O}(\gamma C \log_2 C) \). The communication complexity comprises uploading encrypted model parameters and participating in the initial registration phase, totaling \( \mathcal{O}(\omega C + 2C) \). These costs are effectively mitigated by redundancy-aware optimizations, which reduce the volume of redundant data transmitted from client to server.

\subsection*{b) Auditor's Computational and Communication Cost}

The Auditor's primary tasks include anomaly detection using GMM and MD, incurring a computational cost of \( \mathcal{O}(C G \delta + C \gamma \log \gamma) \), where \( G \) represents the number of Gaussian clusters and \( \delta \) the feature space dimensionality. Communication overhead arises from interaction with both server and clients, expressed as \( \mathcal{O}(R(\Psi + C)) \), where \( R \) is the number of audit rounds, \( \Psi \leq \Theta \) is the number of gradient evaluations, and \( \Theta \) is the total reliability checks. Selective auditing, enabled by our optimization framework, reduces both computation and communication loads without sacrificing detection performance.

\begin{table*}[h!]
\centering
\caption{Computational Cost Before and After Adversarial Client Removal}
\label{tab:attack_cost}
\begin{tabular}{lccc}
\toprule
\textbf{FL Component} & \textbf{Before Adversary Removal} & \textbf{After Adversary Removal} & \textbf{Computational Gain} \\
\midrule
Client \( c_j \) 
    & \( \mathcal{O}(\gamma C \log_2 C) \) 
    & \( \mathcal{O}(\gamma (C-A) \log_2 \left(\frac{C-A}{2}\right)) \) 
    & \( A \log_2 \left(\frac{C}{C-A} \right) \) \\
Auditor 
    & \( \mathcal{O}(C G \delta + C \gamma \log \gamma) \) 
    & \( \mathcal{O}((C - A) G \delta + (C - A) \gamma \log \gamma) \) 
    & \( \frac{A}{C} (G \delta + \gamma \log \gamma) \) \\
Server 
    & \( \mathcal{O}(\omega C \log C) \) 
    & \( \mathcal{O}(\omega (C - A) \log (C - A)) \) 
    & \( \frac{A}{C} \log \left(\frac{C}{C - A} \right) \) \\
\bottomrule
\end{tabular}
\end{table*}

\subsection*{c) Server-Side Complexity Analysis}

The central server aggregates encrypted updates from all clients with a computational complexity of \( \mathcal{O}(\omega C \log C) \). Its communication cost is defined as \( \mathcal{O}(\Phi + CV) \), where \( \Phi \) and \( V \) represent interactions with the Auditor and clients, respectively. The Auditor filters malicious or redundant updates before aggregation, thus reducing the server’s computational load and expediting training convergence. Quantitative gains in efficiency following adversarial user elimination are summarized in Tables~\ref{tab:attack_cost} and~\ref{tab:overhead}. The Auditor's overhead is shown to be minimal, especially in low-threat environments, making it a justifiable trade-off for enhanced system resilience.

\subsection*{d) Cryptographic Cost and Execution Time}

Given the security limitations of legacy key sizes in Paillier-based cryptosystems, this work adopts larger keys conforming to current security standards. The trade-off between computational cost and security is carefully balanced. As key length increases, encryption latency and memory consumption rise, as demonstrated in Table~\ref{tab:keysize_time}. Furthermore, as the number of encrypted gradient elements increases, computational overhead scales accordingly. Thus, the model dynamically tunes key size and batch size to balance privacy with performance, ensuring robust security with feasible latency.





\begin{table*}[h!]
\centering
\caption{Communication Overhead Before and After Adversarial Client Removal}
\label{tab:overhead}
\begin{tabular}{lccc}
\toprule
\textbf{FL Component} & \textbf{Before Adversary Removal} & \textbf{After Adversary Removal} & \textbf{Communication Reduction} \\
\midrule
Client \( c_j \) 
    & \( \mathcal{O}(\omega C + C^2) \) 
    & \( \mathcal{O}(\omega (C - A) + (C - A)^2) \) 
    & \( \frac{A}{C} \left(\frac{C + A}{C} \right) \) \\
Auditor 
    & \( \mathcal{O}(R(\Psi + C)) \) 
    & \( \mathcal{O}(R(\Psi + (C - A))) \) 
    & \( \frac{A}{C} \left(1 - \frac{\Psi}{C} \right) \) \\
Server 
    & \( \mathcal{O}(\Phi + C V) \) 
    & \( \mathcal{O}(\Phi + (C - A) V) \) 
    & \( \frac{A}{C} \left(\frac{C - V}{V} \right) \) \\
\bottomrule
\end{tabular}
\end{table*}

\begin{table*}[htbp]
\centering
\caption{Cryptographic Operation Times Across Key Sizes}
\label{tab:keysize_time}
\begin{tabular}{cccc}
\toprule
\textbf{Key Size (bits)} & \textbf{Encryption Time} & \textbf{Decryption Time} & \textbf{Key Generation Time} \\
\midrule
128  & 0.000114 s (8771 ops/s) & 0.000072 s (13919 ops/s) & 0.001 s \\
256  & 0.000399 s (2508 ops/s) & 0.000171 s (5847 ops/s)  & 0.030 s \\
512  & 0.002450 s (408 ops/s)  & 0.000783 s (1277 ops/s)  & 0.030 s \\
1024 & 0.013911 s (71 ops/s)   & 0.004607 s (217 ops/s)   & 0.200 s \\
2048 & 0.098179 s (10 ops/s)   & 0.028408 s (35 ops/s)    & 0.430 s \\
\bottomrule
\end{tabular}
\end{table*}

\subsection{Scalability to Large-Scale IoT Deployments}

The proposed framework is primarily designed for cross-silo FL environments, where the number of clients is relatively small and each participant possesses sufficient computational resources. Accordingly, our experiments consider 20 clients, all of which participate in every communication round. However, in practical large-scale IoT deployments, where thousands of edge devices may be involved, this full participation model becomes infeasible due to increased communication overhead and computational burden on the central server. In such scenarios, the server must wait for all clients to upload their local models before aggregation. This dependency introduces a bottleneck, especially if some clients are slow (i.e., stragglers) or unable to complete their local training within the prescribed time. To address these challenges, a *client selection mechanism* can be integrated, wherein only a randomly chosen subset of clients participates in each communication round. This approach, widely adopted in FL literature, significantly reduces server load and improves training efficiency. Furthermore, to mitigate the issue of straggling clients, the system can incorporate timeout strategies to exclude delayed responses from the aggregation process. Such client drop strategies prevent indefinite server blocking and maintain the momentum of the learning process. Both client selection and straggler mitigation are active areas of research in scalable federated learning, and extending our framework to support these mechanisms is a promising avenue for future work.

\subsection{Incorporating Recent Advances in Neural Architectures}

Selecting an appropriate neural network architecture is crucial for optimizing the performance of IDS. In this study, we employ a one-dimensional convolutional neural network (1D-CNN) as the local model due to its proven efficacy in modeling time-series and sequential data typical in network traffic analysis. While our current model achieves competitive detection performance, exploring advanced neural architectures could further enhance its predictive capability. Recent studies have demonstrated the effectiveness of hybrid architectures in capturing both spatial and temporal dependencies in sequential data. For instance, the work in \cite{tsokov2022hybrid} proposed a deep spatiotemporal model that integrates two-dimensional CNNs with Long Short-Term Memory (LSTM) networks. In this hybrid model, CNNs extract spatial features while LSTMs model temporal correlations, thereby improving detection performance on sequential datasets. Similarly, the approach in \cite{feng2023tensor} introduced a differentially private tensor-based recurrent neural network tailored for IoT environments. This model not only demonstrated robust detection capabilities but also ensured privacy preservation through the integration of differential privacy mechanisms. Both aforementioned models were deployed in decentralized settings and demonstrate significant potential for adaptation to FL scenarios. Incorporating such architectures into our federated intrusion detection framework could yield enhanced performance and improved privacy guarantees. Future research can focus on the integration and empirical evaluation of these advanced models within the federated IDS paradigm.

\section{Conclusion}
\label{conclusion}
In this study, we proposed SecureDyn-FL, a robust FL-based IDS designed to tackle the challenges of poisoning attacks and non-IID data distributions in IoT environments. Our empirical findings confirm that malicious clients can significantly degrade the performance of federated IDS models, with the impact being more severe under non-IID data settings, a scenario commonly encountered in real-world IoT deployments due to data heterogeneity and imbalance.  To mitigate these challenges, we introduced a personalized FL approach that enables local model customization, allowing each client to adapt to its unique data distribution while still contributing to the global model. This enhances resilience against performance degradation caused by non-IID data. Additionally, we implemented a server-side malicious client detection mechanism that uses anomaly detection to identify and exclude adversarial updates, protecting the global model from poisoning attacks. Extensive experiments on the N-BaIoT dataset confirm the effectiveness of SecureDyn-FL. Evaluated using metrics such as accuracy, F1-score, and detection rate, our approach consistently outperforms state-of-the-art baseline methods under both IID and non-IID conditions. Notably, SecureDyn-FL shows strong robustness against poisoning attacks, including label-flipping and model update poisoning, while maintaining high detection performance. These results highlight its practical value for secure and scalable intrusion detection in heterogeneous IoT networks. 

Future work will explore advanced neural architectures (e.g., graph neural networks and attention mechanisms) to capture IoT data complexities, scalable client selection strategies (e.g., reinforcement learning methods), and enhanced attack detection techniques incorporating differential privacy and secure multi-party computation. In addition, we recognize that blockchain and distributed ledger technologies (DLTs) can provide immutable logging of model updates, decentralized trust management, and verifiable audit trails. Integrating such mechanisms with SecureDyn-FL represents a promising direction to further strengthen the trustworthiness and accountability of federated learning in IoT networks. Collectively, these efforts aim to improve the adaptability, robustness, and transparency of SecureDyn-FL in large-scale federated environments.

\ifCLASSOPTIONcaptionsoff
  \newpage

\fi

\bibliographystyle{IEEEtran} 
\bibliography{References} 

@article{hojlo2021future,
  title={Future of industry ecosystems: Shared data and insights},
  author={Hojlo, Jeffrey},
  journal={IDC Ind. Rep},
  year={2021}
}

@inproceedings{antonakakis2017understanding,
  title={Understanding the mirai botnet},
  author={Antonakakis, Manos and April, Tim and Bailey, Michael and Bernhard, Matt and Bursztein, Elie and Cochran, Jaime and Durumeric, Zakir and Halderman, J Alex and Invernizzi, Luca and Kallitsis, Michalis and others},
  booktitle={26th USENIX security symposium (USENIX Security 17)},
  pages={1093--1110},
  year={2017}
}

@article{agrawal2022federated,
  title={Federated learning for intrusion detection system: Concepts, challenges and future directions},
  author={Agrawal, Shaashwat and Sarkar, Sagnik and Aouedi, Ons and Yenduri, Gokul and Piamrat, Kandaraj and Alazab, Mamoun and Bhattacharya, Sweta and Maddikunta, Praveen Kumar Reddy and Gadekallu, Thippa Reddy},
  journal={Computer Communications},
  volume={195},
  pages={346--361},
  year={2022},
  publisher={Elsevier}
}

@inproceedings{mcmahan2017communication,
  title={Communication-efficient learning of deep networks from decentralized data},
  author={McMahan, Brendan and Moore, Eider and Ramage, Daniel and Hampson, Seth and y Arcas, Blaise Aguera},
  booktitle={Artificial intelligence and statistics},
  pages={1273--1282},
  year={2017},
  organization={PMLR}
}

@article{ferrag2021federated,
  title={Federated deep learning for cyber security in the internet of things: Concepts, applications, and experimental analysis},
  author={Ferrag, Mohamed Amine and Friha, Othmane and Maglaras, Leandros and Janicke, Helge and Shu, Lei},
  journal={IEEe Access},
  volume={9},
  pages={138509--138542},
  year={2021},
  publisher={IEEE}
}

@article{mothukuri2021survey,
  title={A survey on security and privacy of federated learning},
  author={Mothukuri, Viraaji and Parizi, Reza M and Pouriyeh, Seyedamin and Huang, Yan and Dehghantanha, Ali and Srivastava, Gautam},
  journal={Future Generation Computer Systems},
  volume={115},
  pages={619--640},
  year={2021},
  publisher={Elsevier}
}

@article{rey2022federated,
  title={Federated learning for malware detection in IoT devices},
  author={Rey, Valerian and S{\'a}nchez, Pedro Miguel S{\'a}nchez and Celdr{\'a}n, Alberto Huertas and Bovet, G{\'e}r{\^o}me},
  journal={Computer Networks},
  volume={204},
  pages={108693},
  year={2022},
  publisher={Elsevier}
}

@article{zhang2022secfednids,
  title={SecFedNIDS: Robust defense for poisoning attack against federated learning-based network intrusion detection system},
  author={Zhang, Zhao and Zhang, Yong and Guo, Da and Yao, Lei and Li, Zhao},
  journal={Future Generation Computer Systems},
  volume={134},
  pages={154--169},
  year={2022},
  publisher={Elsevier}
}

@article{tt3,
  title={Overview of {RIS}-enabled secure transmission in {6G} wireless networks},
  author={Bae, JungSook and Khalid, Waqas and Lee, Anseok and Lee, Heesoo and Noh , Song and Yu, Heejung},
  journal={Digital Communications and Networks},
  volume={10},
  pages={1553--1565},
  year={2024},
  publisher={Elsevier}
}

@article{tt4,
  title={Effects of co-channel interference on {RIS} empowered wireless networks amid multiple eavesdropping attempts},
  author={Ruku, Md Roisul Ajom and Ibrahim, Md and Badrudduza, ASM and Ansari, Imran Shafique and Khalid, Waqas and Yu, Heejung},
  journal={ICT Express},
  volume={10},
  pages={491--497},
  year={2024},
  publisher={Elsevier}
}

@inproceedings{fang2020local,
  title={Local model poisoning attacks to $\{$Byzantine-Robust$\}$ federated learning},
  author={Fang, Minghong and Cao, Xiaoyu and Jia, Jinyuan and Gong, Neil},
  booktitle={29th USENIX security symposium (USENIX Security 20)},
  pages={1605--1622},
  year={2020}
}

@article{chang2023privacy,
  title={Privacy-preserving federated learning via functional encryption, revisited},
  author={Chang, Yansong and Zhang, Kai and Gong, Junqing and Qian, Haifeng},
  journal={IEEE Transactions on Information Forensics and Security},
  volume={18},
  pages={1855--1869},
  year={2023},
  publisher={IEEE}
}

@article{xu2019verifynet,
  title={VerifyNet: Secure and verifiable federated learning},
  author={Xu, Guowen and Li, Hongwei and Liu, Sen and Yang, Kan and Lin, Xiaodong},
  journal={IEEE Transactions on Information Forensics and Security},
  volume={15},
  pages={911--926},
  year={2019},
  publisher={IEEE}
}

@article{liu2021privacy,
  title={Privacy-enhanced federated learning against poisoning adversaries},
  author={Liu, Xiaoyuan and Li, Hongwei and Xu, Guowen and Chen, Zongqi and Huang, Xiaoming and Lu, Rongxing},
  journal={IEEE Transactions on Information Forensics and Security},
  volume={16},
  pages={4574--4588},
  year={2021},
  publisher={IEEE}
}

@article{ma2022shieldfl,
  title={ShieldFL: Mitigating model poisoning attacks in privacy-preserving federated learning},
  author={Ma, Zhuoran and Ma, Jianfeng and Miao, Yinbin and Li, Yingjiu and Deng, Robert H},
  journal={IEEE Transactions on Information Forensics and Security},
  volume={17},
  pages={1639--1654},
  year={2022},
  publisher={IEEE}
}

@article{cao2020fltrust,
  title={Fltrust: Byzantine-robust federated learning via trust bootstrapping},
  author={Cao, Xiaoyu and Fang, Minghong and Liu, Jia and Gong, Neil Zhenqiang},
  journal={arXiv preprint arXiv:2012.13995},
  year={2020}
}

@inproceedings{shen2016auror,
  title={Auror: Defending against poisoning attacks in collaborative deep learning systems},
  author={Shen, Shiqi and Tople, Shruti and Saxena, Prateek},
  booktitle={Proceedings of the 32nd annual conference on computer security applications},
  pages={508--519},
  year={2016}
}

@article{blanchard2017machine,
  title={Machine learning with adversaries: Byzantine tolerant gradient descent},
  author={Blanchard, Peva and El Mhamdi, El Mahdi and Guerraoui, Rachid and Stainer, Julien},
  journal={Advances in neural information processing systems},
  volume={30},
  year={2017}
}

@article{xu2022hercules,
  title={Hercules: Boosting the performance of privacy-preserving federated learning},
  author={Xu, Guowen and Han, Xingshuo and Xu, Shengmin and Zhang, Tianwei and Li, Hongwei and Huang, Xinyi and Deng, Robert H},
  journal={IEEE Transactions on Dependable and Secure Computing},
  volume={20},
  number={5},
  pages={4418--4433},
  year={2022},
  publisher={IEEE}
}

@article{zhao2019privacy,
  title={Privacy-preserving collaborative deep learning with unreliable participants},
  author={Zhao, Lingchen and Wang, Qian and Zou, Qin and Zhang, Yan and Chen, Yanjiao},
  journal={IEEE Transactions on Information Forensics and Security},
  volume={15},
  pages={1486--1500},
  year={2019},
  publisher={IEEE}
}

@article{sebert2022protecting,
  title={Protecting data from all parties: Combining FHE and DP in federated learning},
  author={S{\'e}bert, Arnaud Grivet and Sirdey, Renaud and Stan, Oana and Gouy-Pailler, C{\'e}dric},
  journal={arXiv preprint arXiv:2205.04330},
  year={2022}
}

@article{xu2022privacy,
  title={Privacy-preserving decentralized deep learning with multiparty homomorphic encryption},
  author={Xu, Guowen and Li, Guanlin and Guo, Shangwei and Zhang, Tianwei and Li, Hongwei},
  journal={arXiv preprint arXiv:2207.04604},
  year={2022}
}

@article{sav2022privacy,
  title={Privacy-preserving federated recurrent neural networks},
  author={Sav, Sinem and Diaa, Abdulrahman and Pyrgelis, Apostolos and Bossuat, Jean-Philippe and Hubaux, Jean-Pierre},
  journal={arXiv preprint arXiv:2207.13947},
  year={2022}
}

@article{sav2020poseidon,
  title={POSEIDON: Privacy-preserving federated neural network learning},
  author={Sav, Sinem and Pyrgelis, Apostolos and Troncoso-Pastoriza, Juan R and Froelicher, David and Bossuat, Jean-Philippe and Sousa, Joao Sa and Hubaux, Jean-Pierre},
  journal={arXiv preprint arXiv:2009.00349},
  year={2020}
}

@inproceedings{zhang2020enabling,
  title={Enabling execution assurance of federated learning at untrusted participants},
  author={Zhang, Xiaoli and Li, Fengting and Zhang, Zeyu and Li, Qi and Wang, Cong and Wu, Jianping},
  booktitle={IEEE INFOCOM 2020-IEEE Conference on Computer Communications},
  pages={1877--1886},
  year={2020},
  organization={IEEE}
}

@inproceedings{gehlhar2023safefl,
  title={SafeFL: MPC-friendly framework for private and robust federated learning},
  author={Gehlhar, Till and Marx, Felix and Schneider, Thomas and Suresh, Ajith and Wehrle, Tobias and Yalame, Hossein},
  booktitle={2023 IEEE Security and Privacy Workshops (SPW)},
  pages={69--76},
  year={2023},
  organization={IEEE}
}

@inproceedings{xu2021else,
  title={What else is leaked when eavesdropping federated learning?},
  author={Xu, Chuan and Neglia, Giovanni},
  booktitle={CCS workshop Privacy Preserving Machine Learning (PPML)},
  year={2021}
}

@inproceedings{driouich2022novel,
  title={A novel model-based attribute inference attack in federated learning},
  author={Driouich, Ilias and Xu, Chuan and Neglia, Giovanni and Giroire, Frederic and Thomas, Eoin},
  booktitle={FL-NeurIPS'22-Federated Learning: Recent Advances and New Challenges workshop in Conjunction with NeurIPS 2022},
  year={2022}
}

@article{wang2019eavesdrop,
  title={Eavesdrop the composition proportion of training labels in federated learning},
  author={Wang, Lixu and Xu, Shichao and Wang, Xiao and Zhu, Qi},
  journal={arXiv preprint arXiv:1910.06044},
  year={2019}
}

@inproceedings{wan2021shielding,
  title={Shielding federated learning: A new attack approach and its defense},
  author={Wan, Wei and Lu, Jianrong and Hu, Shengshan and Zhang, Leo Yu and Pei, Xiaobing},
  booktitle={2021 IEEE wireless communications and networking conference (wcnc)},
  pages={1--7},
  year={2021},
  organization={IEEE}
}

@article{biggio2012poisoning,
  title={Poisoning attacks against support vector machines},
  author={Biggio, Battista and Nelson, Blaine and Laskov, Pavel},
  journal={arXiv preprint arXiv:1206.6389},
  year={2012}
}

@inproceedings{bagdasaryan2020backdoor,
  title={How to backdoor federated learning},
  author={Bagdasaryan, Eugene and Veit, Andreas and Hua, Yiqing and Estrin, Deborah and Shmatikov, Vitaly},
  booktitle={International conference on artificial intelligence and statistics},
  pages={2938--2948},
  year={2020},
  organization={PMLR}
}

@article{soomro2024lightweight,
  title={Lightweight privacy-preserving federated deep intrusion detection for industrial cyber-physical system},
  author={Soomro, Imtiaz Ali and Hussain, Syed Jawad and Ashraf, Zeeshan and Alnfiai, Mrim M and Alotaibi, Nouf Nawar and others},
  journal={Journal of Communications and Networks},
  volume={26},
  number={6},
  pages={632--649},
  year={2024},
  publisher={KICS}
}

@article{zhu2021distributed,
  title={Distributed additive encryption and quantization for privacy preserving federated deep learning},
  author={Zhu, Hangyu and Wang, Rui and Jin, Yaochu and Liang, Kaitai and Ning, Jianting},
  journal={Neurocomputing},
  volume={463},
  pages={309--327},
  year={2021},
  publisher={Elsevier}
}

@article{meidan2018n,
  title={N-baiot—network-based detection of iot botnet attacks using deep autoencoders},
  author={Meidan, Yair and Bohadana, Michael and Mathov, Yael and Mirsky, Yisroel and Shabtai, Asaf and Breitenbacher, Dominik and Elovici, Yuval},
  journal={IEEE Pervasive Computing},
  volume={17},
  number={3},
  pages={12--22},
  year={2018},
  publisher={IEEE}
}

@article{tsokov2022hybrid,
  title={A hybrid spatiotemporal deep model based on CNN and LSTM for air pollution prediction},
  author={Tsokov, Stefan and Lazarova, Milena and Aleksieva-Petrova, Adelina},
  journal={Sustainability},
  volume={14},
  number={9},
  pages={5104},
  year={2022},
  publisher={MDPI}
}

@article{feng2023tensor,
  title={Tensor recurrent neural network with differential privacy},
  author={Feng, Jun and Yang, Laurence T and Ren, Bocheng and Zou, Deqing and Dong, Mianxiong and Zhang, Shunli},
  journal={IEEE Transactions on Computers},
  volume={73},
  number={3},
  pages={683--693},
  year={2023},
  publisher={IEEE}
}

@article{jebreel2023fl,
  title={Fl-defender: Combating targeted attacks in federated learning},
  author={Jebreel, Najeeb Moharram and Domingo-Ferrer, Josep},
  journal={Knowledge-Based Systems},
  volume={260},
  pages={110178},
  year={2023},
  publisher={Elsevier}
}

@Article{MurtazaIoT,
AUTHOR = {Siddiqi, Murtaza Ahmed and Yu, Heejung and Joung, Jingon},
TITLE = {{5G} Ultra-Reliable Low-Latency Communication Implementation Challenges and Operational Issues with IoT Devices},
JOURNAL = {Electronics},
VOLUME = {8},
YEAR = {2019},
NUMBER = {9},
pages = {981},
DOI = {10.3390/electronics8090981}
}

@article{ZIKRIA2018699,
title = {Internet of Things ({IoT}): Operating System, Applications and Protocols Design, and Validation Techniques},
journal = {Future Generation Computer Systems},
volume = {88},
pages = {699--706},
year = {2018},
doi = {https://doi.org/10.1016/j.future.2018.07.058},
author = {Yousaf Bin Zikria and Heejung Yu and Muhammad Khalil Afzal and Mubashir Husain Rehmani and Oliver Hahm},
}

@article{AI6G_ICTE2022,
title = {Enabling technologies for {AI} empowered {6G} massive radio access networks},
journal = {ICT Express},
year = {2023},
  volume={9},
  number={3},  
  pages={31--355},
month={Jun.},
doi = {10.1016/j.icte.2022.07.002},
author = {Md. Shahjalal and Woojun Kim and Waqas Khalid and Seokjae Moon and Murad Khan and ShuZhi Liu and Suhyeon Lim and Eunjin Kim and Deok-Won Yun and Joohyun Lee and Won-Cheol Lee and Seung-Hoon Hwang and Dongkyun Kim and Jang-Won Lee and Heejung Yu and Youngchul Sung and Yeong Min Jang},
}

@Article{What5G_S2017,
AUTHOR = {Yu, Heejung and Lee, Howon and Jeon, Hongbeom},
TITLE = {What is {5G}? {E}merging {5G} Mobile Services and Network Requirements},
JOURNAL = {Sustainability},
VOLUME = {9},
YEAR = {2017},
NUMBER = {10},
pages = {1848},
DOI = {10.3390/su9101848}
}

@ARTICLE{IOTO_ACCESS18,
  author={Musaddiq, Arslan and Zikria, Yousaf Bin and Hahm, Oliver and Yu, Heejung and Bashir, Ali Kashif and Kim, Sung Won},
  journal={IEEE Access}, 
  title={A Survey on Resource Management in IoT Operating Systems}, 
  year={2018},
  volume={6},
  number={},
  pages={8459-8482},
  doi={10.1109/ACCESS.2018.2808324}
}

@ARTICLE{blockchain_IoT2023,
author={Arshad, Qurat-ul-Ain and Khan, Wazir Zada and Azam, Faisal and Khan, Muhammad Khurram Khan and Yu, Heejung and Zikria, Yousaf Bin},
title={Blockchain-based decentralized trust management in IoT: systems, requirements and challenges},
journal={Complex \& Intelligent Systems},
volume={9}, 
pages={6155–-6176},
year={2023},
doi={10.1007/s40747-023-01058-8}
}

@article{IoTEAAI_2024,
title = {An intelligent IoT intrusion detection system using HeInit-WGAN and SSO-BNMCNN based multivariate feature analysis},
journal = {Engineering Applications of Artificial Intelligence},
volume = {127},
pages = {107132},
year = {2024},
doi = {10.1016/j.engappai.2023.107132},
author = {Wu, Jianbin and Haider, Sami Ahmed and Yu, Heejung and Irshad, Muhammad and Soni, Mukesh and Bhadla, Mohit Kumar and Zikria, Yousaf Bin}
}

@ARTICLE{WiFiCommag2021,
  author={Lee, Il-Gu and Kim, Duk Bai and Choi, Jeongki and Park, Hyungu and Lee, Sok-Kyu and Cho, Juphil and Yu, Heejung},
  journal={IEEE Communications Magazine}, 
  title={WiFi HaLow for Long-Range and Low-Power Internet of Things: System on Chip Development and Performance Evaluation}, 
  year={2021},
  volume={59},
  number={7},
  pages={101-107},
  doi={10.1109/MCOM.001.2000815}
}

@article{asif2024advanced,
  title={Advanced zero-shot learning ({AZSL}) framework for secure model generalization in federated learning},
  author={Asif, Muhammad and Naz, Surayya and Ali, Faheem and Salam, Abdu and Amin, Farhan and Ullah, Faizan and Alabrah, Amerah},
  journal={IEEE Access},
  volume={12},
  number={},
  pages={184393-184407},
  year={2024},
  publisher={IEEE}
}

@inproceedings{erbil2022defending,
  title={Defending against targeted poisoning attacks in federated learning},
  author={Erbil, Pinar and Gursoy, M Emre},
  booktitle={2022 IEEE 4th International Conference on Trust, Privacy and Security in Intelligent Systems, and Applications (TPS-ISA)},
  pages={198--207},
  year={2022},
  organization={IEEE}
}

@inproceedings{gill2023feddefender,
  title={FedDefender: Backdoor attack defense in federated learning},
  author={Gill, Waris and Anwar, Ali and Gulzar, Muhammad Ali},
  booktitle={Proceedings of the 1st International Workshop on Dependability and Trustworthiness of Safety-Critical Systems with Machine Learned Components},
  pages={6--9},
  year={2023}
}

@inproceedings{zhang2022fldetector,
  title={Fldetector: Defending federated learning against model poisoning attacks via detecting malicious clients},
  author={Zhang, Zaixi and Cao, Xiaoyu and Jia, Jinyuan and Gong, Neil Zhenqiang},
  booktitle={Proceedings of the 28th ACM SIGKDD conference on knowledge discovery and data mining},
  pages={2545--2555},
  year={2022}
}

@article{bhavsar2024fl,
  title={Fl-ids: Federated learning-based intrusion detection system using edge devices for transportation iot},
  author={Bhavsar, Mansi H and Bekele, Yohannes B and Roy, Kaushik and Kelly, John C and Limbrick, Daniel},
  journal={IEEe Access},
  volume={12},
  pages={52215--52226},
  year={2024},
  publisher={IEEE}
}

@inproceedings{akinie2025fine,
  title={Fine-Tuning Federated Learning-Based Intrusion Detection Systems for Transportation IoT},
  author={Akinie, Robert and Gyimah, Nana Kankam and Bhavsar, Mansi and Kelly, John},
  booktitle={SoutheastCon 2025},
  pages={1155--1161},
  year={2025},
  organization={IEEE}
}

@article{javeed2024federated,
  title={A federated learning-based zero trust intrusion detection system for Internet of Things},
  author={Javeed, Danish and Saeed, Muhammad Shahid and Adil, Muhammad and Kumar, Prabhat and Jolfaei, Alireza},
  journal={Ad Hoc Networks},
  volume={162},
  pages={103540},
  year={2024},
  publisher={Elsevier}
}

@article{rashid2023federated,
  title={A federated learning-based approach for improving intrusion detection in industrial internet of things networks},
  author={Rashid, Md Mamunur and Khan, Shahriar Usman and Eusufzai, Fariha and Redwan, Md Azharuddin and Sabuj, Saifur Rahman and Elsharief, Mahmoud},
  journal={Network},
  volume={3},
  number={1},
  pages={158--179},
  year={2023},
  publisher={MDPI}
}

@article{ruzafa2021intrusion,
  title={Intrusion detection based on privacy-preserving federated learning for the industrial IoT},
  author={Ruzafa-Alc{\'a}zar, Pedro and Fern{\'a}ndez-Saura, Pablo and M{\'a}rmol-Campos, Enrique and Gonz{\'a}lez-Vidal, Aurora and Hern{\'a}ndez-Ramos, Jos{\'e} L and Bernal-Bernabe, Jorge and Skarmeta, Antonio F},
  journal={IEEE Transactions on Industrial Informatics},
  volume={19},
  number={2},
  pages={1145--1154},
  year={2021},
  publisher={IEEE}
}

@article{hamdi2023federated,
  title={Federated learning-based intrusion detection system for Internet of Things},
  author={Hamdi, Najet},
  journal={International Journal of Information Security},
  volume={22},
  number={6},
  pages={1937--1948},
  year={2023},
  publisher={Springer}
}

@inproceedings{azeez2024federated,
  title={Federated learning for privacy-preserving intrusion detection in IoT networks},
  author={Azeez, Sarmad Dheyaa and Ilyas, Muhammad and Bako, Imad Matti},
  booktitle={2024 International Congress on Human-Computer Interaction, Optimization and Robotic Applications (HORA)},
  pages={1--7},
  year={2024},
  organization={IEEE}
}

@article{mahmud2024privacy,
  title={Privacy-preserving federated learning-based intrusion detection technique for cyber-physical systems},
  author={Mahmud, Syeda Aunanya and Islam, Nazmul and Islam, Zahidul and Rahman, Ziaur and Mehedi, Sk Tanzir},
  journal={Mathematics},
  volume={12},
  number={20},
  pages={3194},
  year={2024},
  publisher={MDPI}
}

@inproceedings{abou2023secure,
  title={Secure and efficient federated learning for robust intrusion detection in IoT networks},
  author={Abou El Houda, Zakaria and Moudoud, Hajar and Khoukhi, Lyes},
  booktitle={GLOBECOM 2023-2023 IEEE Global Communications Conference},
  pages={2668--2673},
  year={2023},
  organization={IEEE}
}

@article{yazdinejad2024robust,
  title={A robust privacy-preserving federated learning model against model poisoning attacks},
  author={Yazdinejad, Abbas and Dehghantanha, Ali and Karimipour, Hadis and Srivastava, Gautam and Parizi, Reza M},
  journal={IEEE Transactions on Information Forensics and Security},
  volume={19},
  pages={6693--6708},
  year={2024},
  publisher={IEEE}
}

@article{miao2024rfed,
  title={RFed: Robustness-enhanced privacy-preserving federated learning against poisoning attack},
  author={Miao, Yinbin and Yan, Xinru and Li, Xinghua and Xu, Shujiang and Liu, Ximeng and Li, Hongwei and Deng, Robert H},
  journal={IEEE Transactions on Information Forensics and Security},
  volume={19},
  pages={5814--5827},
  year={2024},
  publisher={IEEE}
}

@article{wu2025privacy,
  title={Privacy-Preserving Federated Learning Scheme with Mitigating Model Poisoning Attacks: Vulnerabilities and Countermeasures},
  author={Wu, Jiahui and Luo, Fucai and Sun, Tiecheng and Wang, Haiyan and Zhang, Weizhe},
  journal={arXiv preprint arXiv:2506.23622},
  year={2025}
}

@article{ashraf2023robust,
  title={Robust and lightweight symmetric key exchange algorithm for next-generation IoE},
  author={Ashraf, Zeeshan and Sohail, Adnan and Yousaf, Muhammad},
  journal={Internet of Things},
  volume={22},
  pages={100703},
  year={2023},
  publisher={Elsevier}
}

@article{ashraf2023lightweight,
  title={Lightweight Privacy-Preserving Remote User Authentication and Key Agreement Protocol for Next-Generation IoT-Based Smart Healthcare},
  author={Ashraf, Zeeshan and Mahmood, Zahid and Iqbal, Muddesar},
  journal={Future Internet},
  volume={15},
  number={12},
  pages={386},
  year={2023},
  publisher={MDPI}
}

@article{salam2023efficient,
  title={Efficient data collaboration using multi-party privacy preserving machine learning framework},
  author={Salam, Abdu and Abrar, Mohammad and Ullah, Faizan and Khan, Izaz Ahmad and Amin, Farhan and Choi, Gyu Sang},
  journal={IEEE Access},
  volume={11},
  pages={138151--138164},
  year={2023},
  publisher={IEEE}
}

@article{chandu2025federated,
  title={Federated Learning for Distributed IoT Security: A Privacy-Preserving Approach to Intrusion Detection},
  author={Chandu, Gutti and Karthik, Thumula and Parag, Balbudhe},
  journal={IEEE Access},
  year={2025},
  publisher={IEEE}
}

@article{li2023efficient,
  title={An efficient federated learning system for network intrusion detection},
  author={Li, Jianbin and Tong, Xin and Liu, Jinwei and Cheng, Long},
  journal={IEEE Systems Journal},
  volume={17},
  number={2},
  pages={2455--2464},
  year={2023},
  publisher={IEEE}
}

@article{liu2024survey,
  title={A survey on federated learning: a perspective from multi-party computation},
  author={Liu, Fengxia and Zheng, Zhiming and Shi, Yexuan and Tong, Yongxin and Zhang, Yi},
  journal={Frontiers of Computer Science},
  volume={18},
  number={1},
  pages={181336},
  year={2024},
  publisher={Springer}
}

@article{guo2024efficient,
  title={Efficient and privacy-preserving federated learning based on full homomorphic encryption},
  author={Guo, Yuqi and Li, Lin and Zheng, Zhongxiang and Yun, Hanrui and Zhang, Ruoyan and Chang, Xiaolin and Gao, Zhixuan},
  journal={arXiv preprint arXiv:2403.11519},
  year={2024}
}

@article{friha2022felids,
  title={FELIDS: Federated learning-based intrusion detection system for agricultural Internet of Things},
  author={Friha, Othmane and Ferrag, Mohamed Amine and Shu, Lei and Maglaras, Leandros and Choo, Kim-Kwang Raymond and Nafaa, Mehdi},
  journal={Journal of Parallel and Distributed Computing},
  volume={165},
  pages={17--31},
  year={2022},
  publisher={Elsevier}
}

@article{kelli2021ids,
  title={IDS for industrial applications: A federated learning approach with active personalization},
  author={Kelli, Vasiliki and Argyriou, Vasileios and Lagkas, Thomas and Fragulis, George and Grigoriou, Elisavet and Sarigiannidis, Panagiotis},
  journal={Sensors},
  volume={21},
  number={20},
  pages={6743},
  year={2021},
  publisher={MDPI}
}

@article{alsaedi2020ton_iot,
  title={TON\_IoT telemetry dataset: A new generation dataset of IoT and IIoT for data-driven intrusion detection systems},
  author={Alsaedi, Abdullah and Moustafa, Nour and Tari, Zahir and Mahmood, Abdun and Anwar, Adnan},
  journal={Ieee Access},
  volume={8},
  pages={165130--165150},
  year={2020},
  publisher={IEEE}
}

\newpage
\begin{IEEEbiography}[{\includegraphics[width=1 in,height=1.25 in,clip,keepaspectratio]{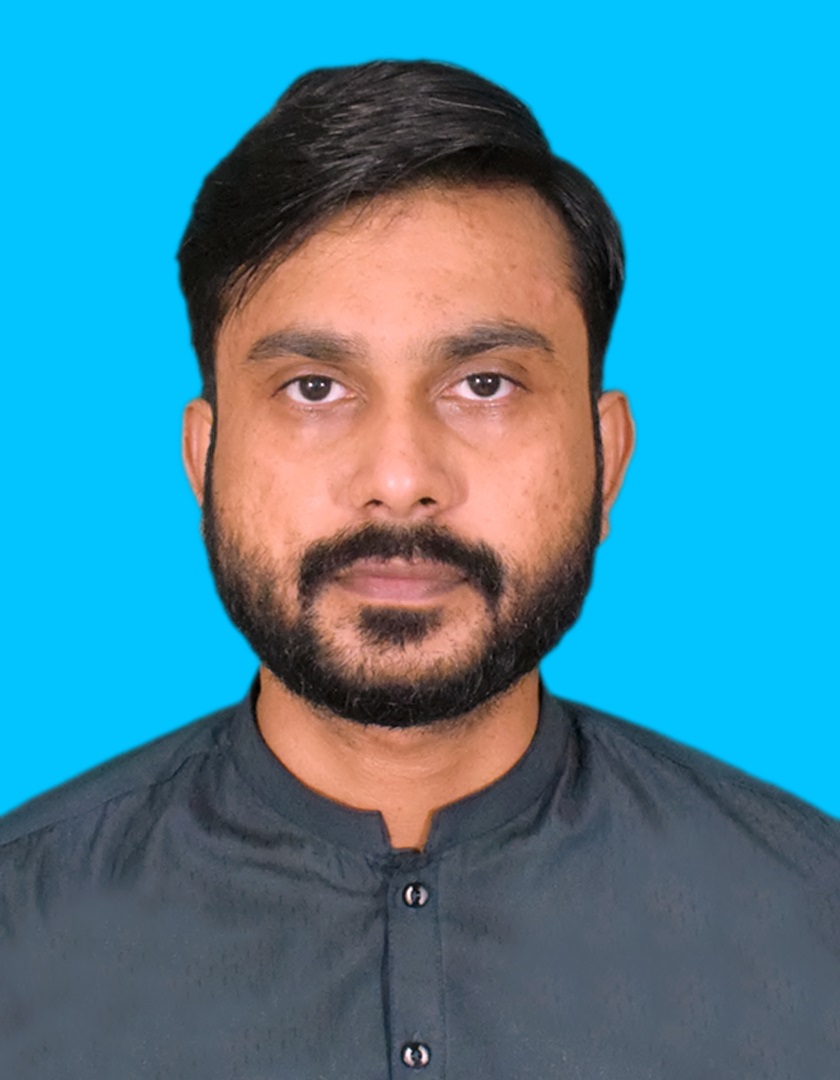}}]{Imtiaz Ali Soomro} is pursuing a PhD in Electrical and Computer Engineering at Sir Syed CASE Institute of Technology, Islamabad, Pakistan. He obtained his B.E. in Telecommunications from Hamdard University, Islamabad, Pakistan 2010. He earned his M.S. in Electrical Engineering with a specialization in Telecom and Networking from COMSATS University, Islamabad, Pakistan, in 2012. His research interests are focused on the application of Federated Learning for IoT, wireless networks, and cybersecurity, particularly on privacy-preserving technologies and secure communication in distributed systems. He is also an IEEE member, actively contributing to the research community through his innovative work in machine learning, IoT, and cybersecurity.
\end{IEEEbiography}
\vspace{11pt}
\begin{IEEEbiography}[{\includegraphics[width=1.05in,height=1.08in,clip]{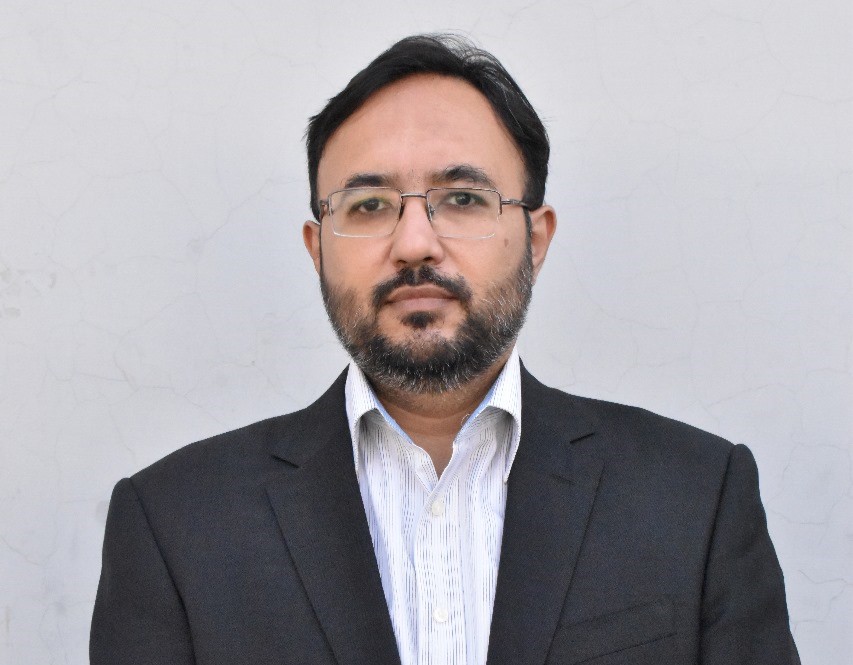}}]{Hamood ur Rehman Khan}
obtained his B.S. degree in electronics engineering from Ghulam Ishaq Khan Institute of Technology in 2000, M.S. degree from the University of Michigan in 2005 and Ph.D. degree from King Fahd University of Petroleum and Minerals in 2019, both in electrical engineering. He was a Senior Member Technical Staff at the Center for Advanced Research and Engineering (CARE), jointly holding appointment as a an Adjunct Professor with the Computer Science Department at Sir Syed-CASE-Institute of Technology. At CARE he has led projects ranging from IoT Platform-as-a-Service systems, cyber-security products, and advanced VLSI based AI platforms for Large Language Model (LLM) inference. Currently, he is an Assistant Professor in the ECE Department of Habib University, teaching various courses pertaining to Electrical and Computer Engineering majors, including, computer architecture, signals and systems, digital communications and statistical inference. In the past, during the period 2000-2003 he was with Avaz Networks Inc, California as a Senior VLSI Design Engineer working on high-density Voice over IP (VoIP) System-on-Chips for gateway media switches. His primary research interests are information theory, signal processing and PHY layer communications for networked systems like WSNs and IoTs.
\end{IEEEbiography}

\vspace{11pt}
\begin{IEEEbiography}[{\includegraphics[width=1in,height=1.25in,clip,keepaspectratio]{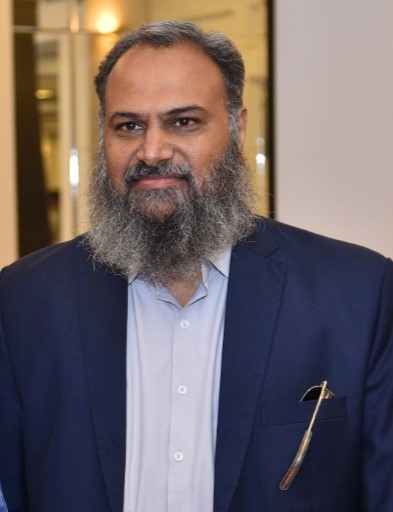}}]{Syed Jawad Hussain}

is an Associate Professor and Chairperson at the Sir Syed Case Institute of Technology, Islamabad, Pakistan. He holds a PhD in Computer Science from Massey University, New Zealand, focusing on developing high-definition video quality experience models. His research interests include Multimedia Communication Networks, Machine Learning, Quality of Service (QoS), Quality of Experience (QoE), Data and Network Security, and Statistical Modeling. Dr. Hussain has extensive experience in academia and industry, having held various leadership roles, including Head of Department positions at institutions in Pakistan and abroad. He has worked on numerous research and consultancy projects, focusing on machine learning, data security, and multimedia communications. Dr. Hussain has published extensively in prestigious journals and conferences, contributing significantly to the field of computer science.
\end{IEEEbiography}

\begin{IEEEbiography}
[{\includegraphics[width=1in,height=1.25in,clip,keepaspectratio]{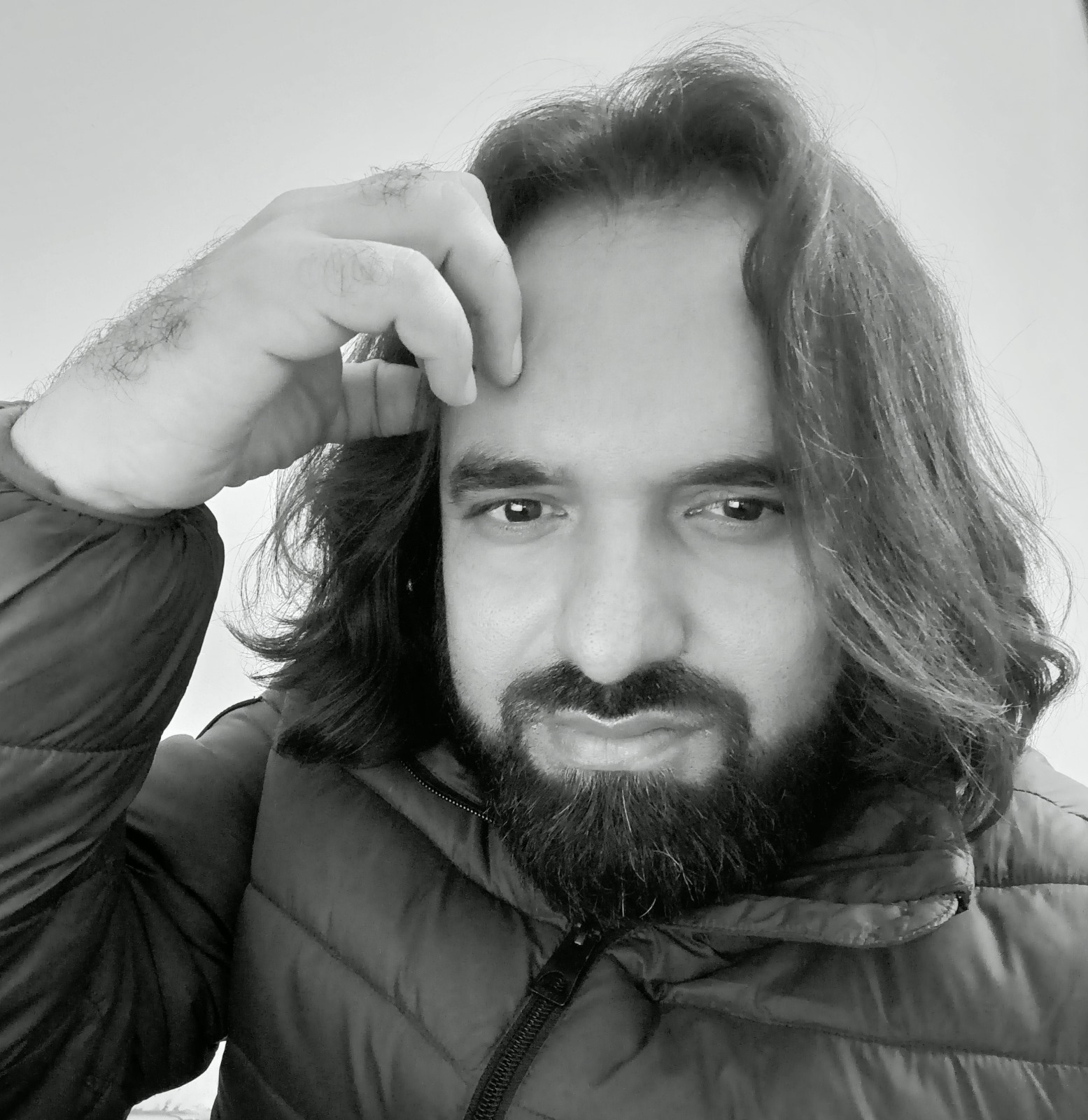}}]
   {Adeel Iqbal} is an Assistant Professor in the School of Computer Science and Engineering at Yeungnam University, South Korea. He specializes in Electrical Engineering. He completed his Bachelor's degree at the FUUAST Islamabad, and later earned his Master's and Ph.D. in Electrical Engineering from COMSATS University Islamabad. Adeel’s research encompasses next-generation cellular networks such as cognitive radio, IoT, D2D, and vehicular systems, along with work in WSNs, machine learning, image processing, and green and renewable energy.

\end{IEEEbiography}

\begin{IEEEbiography}
[{\includegraphics[width=1in,height=1.25in,clip,keepaspectratio]{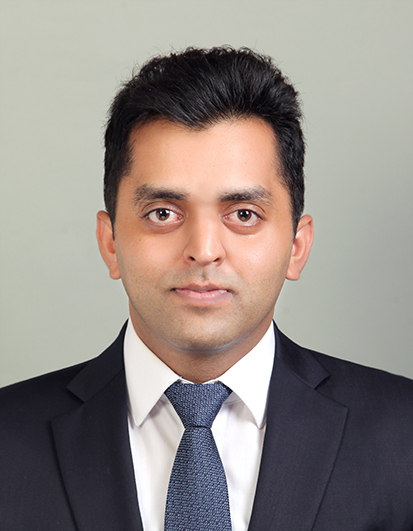}}]
   {Waqas Khalid} (Member, IEEE) received the B.S. degree in electronics engineering from the GIK Institute of Engineering Sciences and Technology, KPK, Pakistan, in 2011, the M.S. degree in information and communication engineering from Inha University, Incheon, South Korea, in 2016, and the Ph.D. degree in information and communication engineering from Yeungnam University, Gyeongsan, South Korea, in 2019. He is currently an Assistant Professor with the Department of Electrical and Electronic Engineering, University of Nottingham Ningbo China (UNNC), Ningbo, China. Previously, he served as a Research Professor at the Institute of Industrial Technology, Korea University, Sejong, South Korea, where he was also the recipient of a National Research Foundation of Korea (NRF) research grant from Jun. 2022 to May 2025. His research interests include physical layer modeling, signal processing, and emerging technologies for 5G/6G networks, including reconfigurable intelligent surfaces, physical-layer security, non-orthogonal multiple access, UAV communications, and the IoTs.
   
\end{IEEEbiography}

\begin{IEEEbiography}
[{\includegraphics[width=1in,height=1.25in,clip,keepaspectratio]{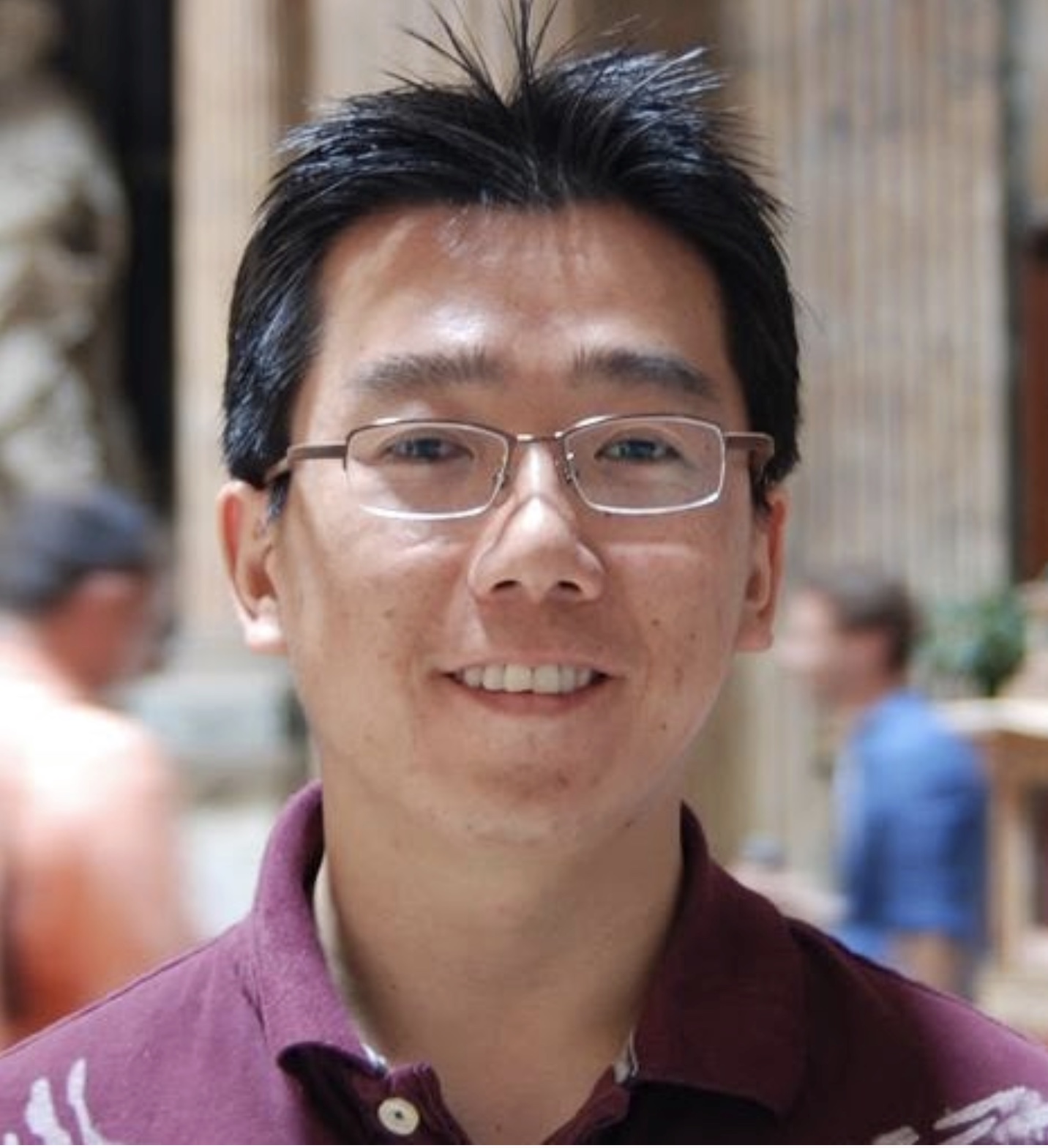}}] {Heejung Yu} (Senior Member, IEEE) received the B.S. degree in radio science and engineering from Korea University, Seoul, South Korea, in 1999, and the M.S. and Ph.D. degrees in electrical engineering from the Korea Advanced Institute of Science and Technology, Daejeon, South Korea, in 2001 and 2011, respectively. From 2001 to 2012, he was with the Electronics and Telecommunications Research Institute, Daejeon, and from 2012 to 2019, he was with Yeungman University, Gyeongsan, South Korea. He is currently a Professor at the Department of Electronics and Information Engineering, Korea University, Sejong, South Korea. His research interests include statistical signal processing and communication theory.
\end{IEEEbiography}

\end{document}